\documentclass[11pt, reqno]{amsart}
\usepackage{amsmath, amsthm}
\usepackage{amscd}
\usepackage{enumerate}
\usepackage{float,amsmath,amssymb,mathrsfs,bm,multirow,graphics}
\usepackage[dvips]{graphicx}
\usepackage[percent]{overpic}
\usepackage[numbers,sort&compress]{natbib}
\usepackage{xcolor}
\usepackage{todonotes}
\usepackage{amsaddr}

\usepackage[latin1]{inputenc}

\allowdisplaybreaks

\newcommand{\R}{{\mathbb{R}}}

\newcommand{\C}{{\mathbb{C}}}
\newcommand{\Z}{{\mathbb{Z}}}

\newcommand{\res}{\text{\upshape Res} \,}
\newcommand{\re}{\text{\upshape Re} \,}
\newcommand{\im}{\text{\upshape Im} \,}

\newcommand{\cut}{\text{cut}}

\newtheorem{theorem}{Theorem}[section]
\newtheorem{proposition}[theorem]{Proposition}
\newtheorem{lemma}[theorem]{Lemma}

\newtheorem{assumption}[theorem]{Assumption}
\newtheorem{remark}[theorem]{Remark}

\newtheorem{RHproblem}[theorem]{RH problem}

\numberwithin{equation}{section}

\usepackage[colorlinks=true]{hyperref}
\hypersetup{urlcolor=blue, citecolor=red, linkcolor=blue}

\begin{document}

\title[A new approach to integrable evolution equations]{A new approach to integrable evolution \\ equations on the circle}

\author{A. S. Fokas}
\address{Department of Applied Mathematics and Theoretical Physics, \\ University of Cambridge, Cambridge CB3 0WA, United Kingdom  \\ and \\ Viterbi School of Engineering, University of Southern California,  \\ Los Angeles, California, 90089-2560, USA.}
\email{t.fokas@damtp.cam.ac.uk} 

\author{J. Lenells}
\address{Department of Mathematics, KTH Royal Institute of Technology, \\ 100 44 Stockholm, Sweden.}
\email{jlenells@kth.se}

\begin{abstract} 
\noindent
We propose a new approach for the solution of initial value problems for integrable evolution equations in the periodic setting based on the unified transform. Using the nonlinear Schr\"odinger equation as a model example, we show that the solution of the initial value problem on the circle can be expressed in terms of the solution of a Riemann-Hilbert problem whose formulation involves quantities which are defined in terms of the initial data alone. Our approach provides an effective solution of the problem on the circle which is conceptually analogous to the solution of the problem on the line via the inverse scattering transform. 
\end{abstract}

\maketitle

\noindent
{\small{\sc AMS Subject Classification (2020)}: 	35Q55, 37K15, 35G30.}

\noindent
{\small{\sc Keywords}:  Integrable evolution equation, periodic solution, Riemann-Hilbert problem, finite-gap solution, inverse scattering, unified transform method, Fokas method, linearizable boundary condition.}

\setcounter{tocdepth}{1}
\tableofcontents

\section{Introduction}
Following the seminal discovery that the Korteweg--de Vries (KdV) equation can be solved analytically via a novel methodology \cite{GGKM1967}, Peter Lax understood that the distinguished feature of this equation was that it can be written as the compatibility condition of two linear eigenvalue equations, later called a Lax pair. 
An explosion of results regarding equations possessing a Lax pair, later called integrable, occurred after the decisive work of Zakharov and Shabat \cite{ZS1972}: the nonlinear Schr\"odinger (NLS) equation, which is  another `generic' equation of physical significance, namely, an equation which is derived under natural asymptotic considerations from a large class of PDEs,  can also be linearized via the analysis of its associated Lax pair. In this way, a new method in mathematical physics was born known as the inverse scattering transform. This method, which as clearly understood by Mark Ablowitz et al \cite{AKNS1974} (see also \cite{ZM1974, ZS1974}) can be thought of as the implementation of a `nonlinear Fourier transform', consists of two steps: (1) The first step, often referred to as the solution of the direct problem, involves the construction of the so-called scattering data, usually denoted by $a(k)$ and $b(k)$. These functions, which are defined in the spectral (Fourier) space  are expressed in terms of the initial datum via the solution of a linear Volterra integral equation. In the linear limit where the solution $q(x ,t)$ of the associated nonlinear PDE is assumed to be small, $a(k)$ tends to $1$ and $b(k)$ tends to the usual Fourier transform of the initial datum, $q_0(x)=q(x,0)$. (2)  The second step, often referred to as the solution of the inverse problem, involves the construction of the solution $q(x, t)$ in terms of time-dependent scattering data. By employing the $t$-part of the Lax pair it can be shown that the function $a(k)$ is time independent, whereas the time dependence of $b(k, t)$ is simple, namely it is the same as the time dependence of the underlying linear Fourier transform. The solution $q(x, t)$ can be expressed via the solution of a linear integral equation of the Fredholm type. In the case of the KdV, this is the so-called Gelfand--Levitan--Marchenko (GLM) equation, first obtained in connection with the scattering theory of the time-independent Schr\"odinger operator. Although it was realized by Zakharov and Shabat that the inverse problem of NLS can also be formulated as a classical problem in complex analysis called a Riemann-Hilbert (RH) problem, the GLM  equation continued to dominate the theory of integrable systems until the works of one of the authors and Ablowitz: it was shown in \cite{FA1983} that the inverse problem associated with the Benjamin--Ono equation gives rise to a RH problem, which in contrast to the local nature of the RH problems arising in the usual integrable evolution equations in one space variable, is  nonlocal. Actually, a nonlocal RH problem also characterizes the solution of the inverse problems associated with the initial value problem of many integrable nonlinear evolution equations in two space dimensions, including KPI \cite{FA1983b}, DSI \cite{FS1989}, and the N-wave interactions \cite{SF1991}. The formulation of the inverse problems in terms of either local or nonlocal RH problems made it clear that the essence of the underlying mathematical structure relevant to the solution of these problems is the following:  there exist eigenfunctions of the associated $t$-independent part of the Lax pair which are sectionally holomorphic, namely, they are holomorphic in different domains of the complex-plane.  Interestingly, there exist many nonlinear integrable evolution PDEs in two space dimensions, like KPII and DSII, whose associated eigenfunctions are nowhere analytic in the complex k-plane. The analysis of these equations requires the formulation of a so-called d-bar problem, instead of a RH problem. The d-bar methodology was introduced in the area of integrable systems by Beals and Coifman \cite{BC1985} who employed it for equations in one space dimensions (for which the RH formulation is actually preferable). The d-bar formulation was used in two space dimensions where its use is indispensable, by one of the authors, Ablowitz, and others \cite{FA1983c, ABF1983, BC1986, BC1989, FS1992}. Following the above developments, it became clear that the initial value problem for nonlinear evolution PDEs in one and two space variables can be solved by employing a local RH formalism and either a nonlocal RH or a d-bar formalism, respectively. Furthermore, it was shown by one of the authors and Gelfand that Fourier transforms in one and two space dimensions can also be derived via a RH and a d-bar formalism, respectively \cite{FG1994}. Hence, the initial value problem for linear and for integrable nonlinear evolution equations in one and two space dimensions can be solved via linear and suitable nonlinear Fourier transforms, which can be constructed via RH and d-bar formalisms.

Solving one-dimensional evolution equations formulated on the half-line or on an interval is far more challenging than solving the associated initial value problem. The first such problem to be analyzed for nonlinear integrable equations was the so-called periodic problem, namely the problem formulated on the finite interval $0<x<L$, with $x$-periodic initial conditions. In this direction, remarkable results were obtained by many authors using techniques of algebraic geometry. 
For almost 100 years the only known $x$-periodic solution of the KdV equation was the periodic analogue of the one-soliton solution known as the cnoidal wave solution, obtained by Korteweg and de Vries in 1895. The possibility of constructing the $x$-periodic analogue of the multi-soliton solution of the KdV equation became clear to Russian mathematicians in 1973-1974, following the discovery of the work of Akhiezer \cite{A1961} by Matveev,  and the pioneering ideas of  Sergei Novikov \cite{N1974}. The first explicit formulas expressed in terms of theta functions were obtained by Alexander Its and Matveev \cite{IM1975}, and then by Boris Dubrovin \cite{D1975}. The first explicit expression for the associated eigenfunction obtained by Its was never published but was reported in \cite{DMN1976}. Parallel developments  took place in USA with the works of Kac and van Moerbeke \cite{KvM1975}, Lax \cite{L1975}, McKean and van Moerbeke \cite{MvM1975}, and Flaschka and McLaughlin \cite{FM1976}. The NLS was analyzed in \cite{I1976} and in \cite{IK1976}. The extension of the above results to multidimensions was achieved by Igor Krichever \cite{K1975}. The relevant approach was later called the Baker--Akhiezer formalism after the realization that some of the key mathematical structures needed for the finite-gap integration were introduced by Baker in 1897 and 1907 \cite{B1928}. Hamiltonian aspects of the finite-gap formalism were developed by Novikov, Dubrovin, Bogoyavlenskij, and Gelfand and Dickey (see the collection of articles in \cite{IntegrableSystems1981}). Computational aspects of the finite-gap solutions have been discussed by several investigators; for example, Deconinck et al \cite{DHBvHS2004} and Klein and Richter \cite{KR2005}. An excellent review of the above remarkable developments and their implications in various branches of mathematics and physics is \cite{M2008} (for example, the construction of periodic solutions of the KP equation led to the solution of the famous Schottky problem in algebraic geometry). However, despite the above important results, the solution of the periodic problem for arbitrary initial conditions (as opposed to the particular initial conditions corresponding to the above exact solutions) remained open. The main difficulty of solving this problem using techniques of algebraic geometry is that it requires the construction of Riemann surfaces of infinite genus. Progress in this direction has been recently announced in \cite{MN2019}, where the approach of McKean, Trubowitz, and others is followed, namely a RH problem is  defined on the spectrum of the bands, but now an infinite number of gaps is allowed (this approach would be even more problematic for the focusing NLS). 

Following the solution of the initial value problem in one and two space dimensions, and the construction of  periodic analogues of  multi-soliton solutions, the main open problems in the theory of nonlinear integrable evolution equations, in addition to solving the general periodic problem, became the following: (1) the extension of  integrability to three space dimensions, and (2) the solution of general initial-boundary value problems. Despite the efforts of many researchers, it has not been possible to construct nonlinear integrable evolution equations in $3+1$, i.e. in three space dimensions and one temporal dimension (one of the authors has constructed integrable equations in $4+2$, and has shown that their initial value problem can be solved in terms of a nonlocal d-bar formalism \cite{F2006}; although these equations, at least aesthetically, provide the correct multi-dimensional analogue of the usual integrable equations, the question of reducing them to $3+1$ dimensional equations remains open \cite{FV2018}). Regarding (2) above, after the efforts of several researchers (see for example \cite{FI1996}), a novel new method known as the unified transform (UTM) or the Fokas method, emerged in 1997 \cite{F1997}.

The UTM is based on two novel ideas: (1) Perform the simultaneous spectral analysis of both parts of the Lax pair. This should be contrasted with the inverse scattering transform which analyses only the $t$-independent part. The inverse scattering transform corresponds to implementing a space-variable transform (like the $x$-Fourier transform). Hence, it implements a nonlinear version of the classical idea of the separation of variables, whereas the UTM, by analyzing simultaneously both parts of the Lax pair, goes beyond separation of variables. Indeed, for linear PDEs, the UTM gives rise to a completely new transform based on the synthesis as opposed to separation of variables \cite{F2000, DTV2014, FS2012}. This method has also led to a new rigorous treatment of the question of well-posedness for nonlinear evolution PDEs \cite{HM2018, OY2019}.
(2) The second ingredient of the UTM is the analysis of the global relation; this is an equation coupling appropriate transforms of the initial datum and all boundary values. As an example, in the case of NLS with $0<x< \infty$, the implementation of (1) expresses the solution $q (x, t)$ in terms of a RH problem whose only $(x, t)$ dependence is in the explicit form of $\exp(2ikx +4ik^2t)$ \cite{FIS2005}. The jump matrices of this RH problem depend on the spectral functions $a(k)$, $b(k)$, $A(k)$, and $B(k)$. The functions $a(k)$ and $b(k)$ can be computed in terms of the initial datum $q_0(x)$. However, $A(k)$ and $B(k)$ are defined in terms of $q(0,t)$ and $q_x(0,t)$, and  since only one of these functions  can be specified from the given boundary conditions, the functions $A(k)$ and $B(k)$ cannot be directly determined from the given data.  At this stage, the crucial importance of the global relation becomes evident: it can be used to characterize the unknown boundary value in terms of the given initial and boundary conditions; but, unfortunately this step, in general, is nonlinear \cite{FL2012, LF2015, LF2015b}. However, for a particular class of initial-boundary value problems called linearizable, this nonlinear step can by bypassed. As an example, it is noted that the boundary condition $q_x(0, t)+c q(0,t) = 0$, $c$ real constant, belongs to this class; in this case,  the functions $A(k)$ and $B(k)$ can be expressed in terms of $a(k)$, $b(k)$, and $c$ \cite{FIS2005}.

\subsection{A New Approach to the Periodic Problem}
In 2004, one of the authors and Alexander Its implemented the UTM to NLS on the finite interval $0<x<L$ \cite{FI2004}. In this case, the associated RH problem involves the spectral functions $a(k)$, $b(k)$, $A(k)$, and $B(k)$, mentioned above, as well as the functions $\mathcal{A}(k)$ and $\mathcal{B}(k)$ that are defined in terms of $q(L, t)$ and $q_x(L, t)$. For the periodic problem, $\mathcal{A}(k) = A(k)$ and $\mathcal{B}(k)=B(k)$. It was speculated by the authors of \cite{FI2004} that the periodic problem belongs to the linearizable class. It will be shown here that this is indeed the case: it is possible to construct the solution $q(x, t)$ of the NLS in terms of a RH problem whose jumps are expressed explicitly in terms of the spectral functions $a(k)$ and $b(k)$ (which as noted earlier can be determined in terms of the initial datum $q_0(x)$). In this sense, the periodic problem can be solved with the same level of efficiency as the initial value problem on the infinite line.

In more detail: (a) By using a novel transformation it is possible to map the RH problem formulated in \cite{FI2004} to one whose jump matrices depend on $a(k)$, $b(k)$, and the ratio $\Gamma(k) = B(k)/A(k)$. (b) A general initial-boundary value problem for a one-dimensional evolution equation like the NLS defined on the finite interval $0<x<L$ is formulated for $0<t<T$; although the solution $q(x, t)$ is independent of $T$, the functions $A(k)$ and $B(k)$ do depend on $T$.  It turns out that by  employing a suitable transformation it is possible to map the  RH problem obtained in (a) above to one whose jump matrices depend on $a(k)$, $b(k)$, and the ratio $\tilde{\Gamma}(k)$, where the function $\tilde{\Gamma}(k)$ is independent of $T$. (c) By using the global relation it is possible to express $\tilde{\Gamma}(k)$ in terms of $a(k)$ and $b(k)$. This formula involves a square root, which necessitates the introduction of suitable branch cuts, which in turn introduce additional jumps in the above RH problem. Actually, this square root involves the expression $4 - \Delta(k)^2$, where $\Delta(k)$ denotes the trace of the monodromy matrix. It is well known that this function plays a decisive role in the classical approach of the $x$-periodic problem.

For definiteness, we will use the NLS equation
\begin{align}\label{NLS}
i q_t+q_{xx}-2\lambda q |q|^2=0, \qquad \lambda = \pm 1,
\end{align}
as our model example, but it will be clear that the same steps can be implemented also for other integrable evolution equations, such as the KdV equation. Since the associated RH problems are somewhat different, we will treat both the defocusing (corresponding to $\lambda=1$) and the focusing (corresponding to $\lambda=-1$) versions of (\ref{NLS}).

\subsection{Outline of the paper}
In Section \ref{intervalsec}, we review the application of the unified transform method to the NLS equation posed on a finite interval $[0,L] \subset \R$. In Section \ref{periodicsec}, we restrict our attention to the class of spatially periodic solutions and formulate a RH problem from which the solution of (\ref{NLS}) can be obtained. The solution formula is not yet effective, because the formulation of this RH problem involves a certain function $\Gamma(k)$ which depends on the boundary values $q(0,t) = q(L,t)$ and $q_x(0,t) = q_x(L,t)$. However, in Section \ref{mainsec}, we show that the function $\Gamma(k)$ can be replaced by another function $\tilde{\Gamma}(k)$ without affecting the resulting expression for $q(x,t)$. The definition of $\tilde{\Gamma}(k)$ only involves the initial datum $q(x,0)$, $x \in [0,L]$, and this yields our main result, which is stated in Theorem \ref{mainth}. 
In Section \ref{exponentialexamplesec}, the main result is illustrated by means of an example for which the associated RH problem can be solved explicitly. In the terminology of the finite-gap approach, this example corresponds to one-gap solutions. 

\subsection{Notation}
The four open quadrants of the complex $k$-plane will be denoted by $D_j$, $j = 1,\dots, 4$, and $\Sigma$ will denote the contour $\R \cup i\R$ oriented as in Figure \ref{jumps.pdf}.
The boundary values of a function $f$ on a contour from the left and right will be denoted by $f_+$ and $f_-$, respectively. We will use $\{\sigma_j\}_1^3$ to denote the three standard Pauli matrices. The first and second columns of a $2 \times 2$ matrix $A$  will be denoted by $[A]_1$ and $[A]_2$, respectively.

\begin{figure}
\begin{center}
\begin{overpic}[width=.4\textwidth]{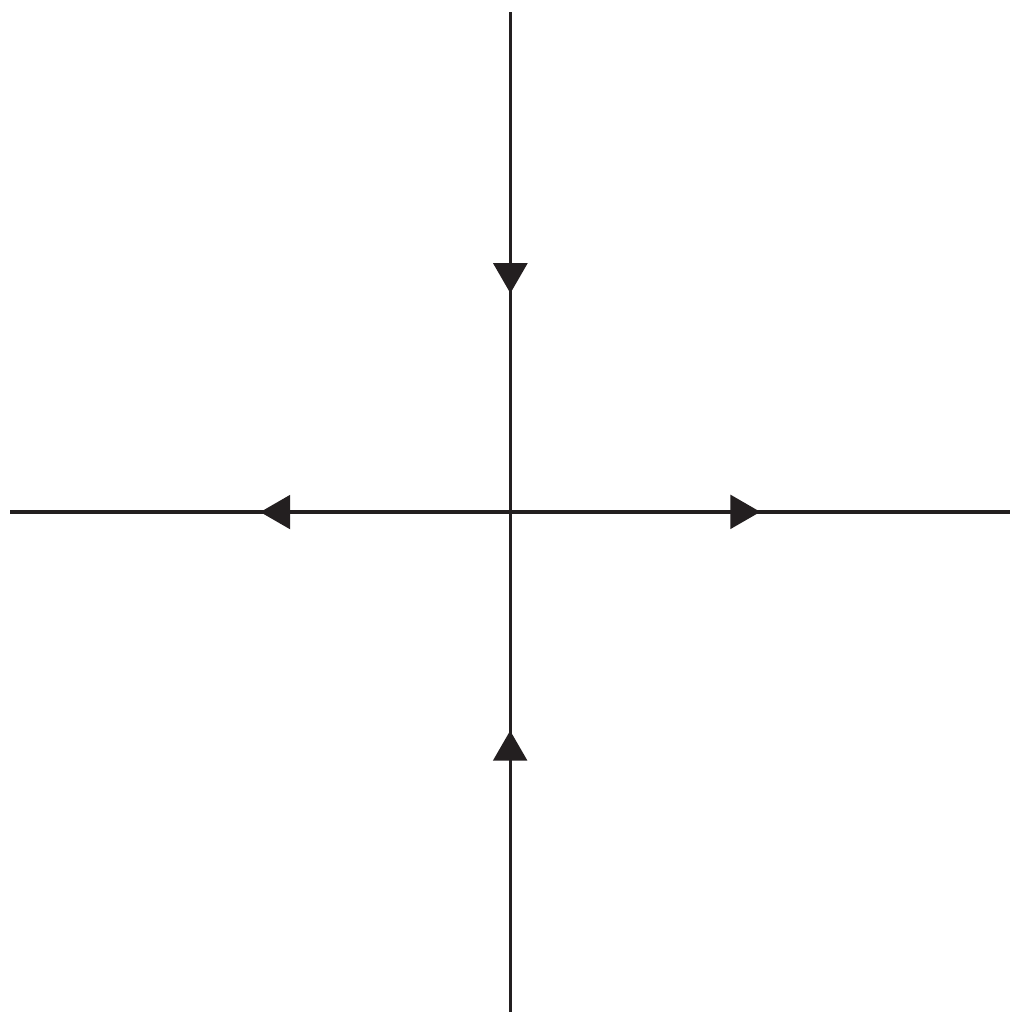}
      \put(74,76){\small $D_1$} 
      \put(20,76){\small $D_2$} 
      \put(20,22){\small $D_3$} 
      \put(74,22){\small $D_4$} 
       \put(70, 53.5){\small $J_2$} 
      \put(52, 71){\small $J_1$} 
      \put(25, 53.5){\small $J_4$} 
      \put(52, 25){\small $J_3$} 
     \end{overpic}
    \caption{\label{jumps.pdf} The four open quadrants $D_1, \dots D_4$, the contour $\Sigma = \R \cup i\R$, and the jump matrices $J_1, \dots, J_4$ defined in (\ref{J1234def}).}
     \end{center}
\end{figure}
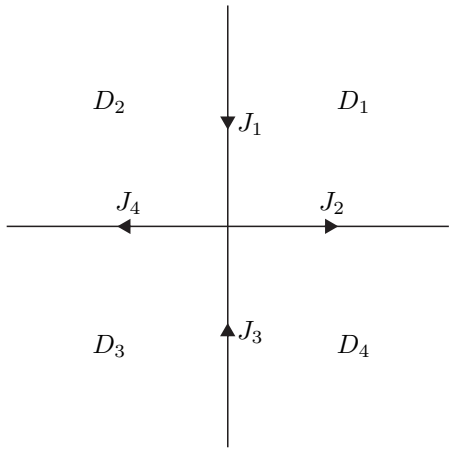

\section{The NLS equation on a finite interval}\label{intervalsec}
Before turning to the periodic problem, we recall some aspects of the analysis of the NLS equation on a finite interval presented in \cite{FI2004} which will be needed in later sections.


The NLS equation (\ref{NLS}) has a Lax pair given by
\begin{align}\label{laxpair}
  \mu_x+ik\hat{\sigma}_3 \mu=Q\mu, \qquad \mu_t+2ik^2 \hat{\sigma}_3 \mu=\tilde{Q}\mu,
\end{align}
where $k \in \C$ is the spectral parameter, $\mu(x,t,k)$ is a $2 \times 2$-matrix valued eigenfunction, the matrices $Q$ and $\tilde{Q}$ are defined in terms of the solution $q(x,t)$ of (\ref{NLS}) by
\begin{align}
Q= \begin{pmatrix} 0 & q \\ \lambda \bar{q} & 0 \end{pmatrix}, \qquad \tilde{Q}=2kQ-i Q_x\sigma_3-i\lambda |q|^2\sigma_3,
\end{align}
and $\hat{\sigma}_3\mu = [\sigma_3, \mu]$.
Suppose $q(x,t)$ is a smooth solution of (\ref{NLS}) defined for  $(x,t) \in [0,L] \times [0,T]$, where $0 < L < \infty$ and $0 < T < \infty$ is some fixed final time.
Following \cite{FI2004}, we let $\mu_j(x,t,k)$, $j = 1,2,3,4$, denote the four solutions of (\ref{laxpair}) which are normalized to be the identity matrix at the points $(0,T)$, $(0,0)$, $(L,0)$, and $(L,T)$, respectively. The spectral functions $a,b,A,B, \mathcal{A}, \mathcal{B}$ are defined for $k \in \C$ by
\begin{align}\label{abdef}
s(k)=\begin{pmatrix} \overline{a(\bar{k})} & b(k) \\ \lambda \overline{b(\bar{k})} & a(k) \end{pmatrix}, \quad
S(k)= \begin{pmatrix} \overline{A(\bar{k})} & B(k) \\ \lambda \overline {B(\bar{k})} & A(k) \end{pmatrix}, \quad 
S_L(k)= \begin{pmatrix} \overline{\mathcal{A}(\bar{k})} & \mathcal{B}(k) \\ \lambda \overline{\mathcal{B}(\bar{k})} & \mathcal{A}(k) \end{pmatrix},
\end{align}
where 
$$s(k)=\mu_3(0,0,k), \qquad S(k)=\mu_1(0,0,k), \qquad S_L(k)=\mu_4(L,0,k).$$ 
Note that $\{\mu_j\}_1^4$, $s$, $S$, and $S_L$ are entire functions of $k$. 
Clearly, $S(k)$ and $S_L(k)$ depend on $T$, whereas $s(k)$ does not. The entries of the matrices $s, S$, and $S_L$ are related by the so-called global relation (see \cite[Eq. (1.4)]{FI2004})
\begin{align}\label{gr}
(a\mathcal{A}+\lambda \bar{b} e^{2ikL}\mathcal{B})B-(b \mathcal{A}+\bar{a} e^{2ikL}\mathcal{B})A=e^{4ik^2T}c^+(k), \qquad k \in \C,
\end{align}
where $c^+(k)$ is an entire function which is of order $1/k$ as $k \to \infty$ in the upper half-plane; in fact
\begin{align}\label{cplusasymptotics}
c^+(k) = O\bigg(\frac{1}{k}\bigg) + O\bigg(\frac{e^{2ikL}}{k}\bigg), \qquad k \to \infty, \ k \in \C.
\end{align}
Here, and in what follows, a bar over a function denotes that the complex conjugate is taken, not only of the function but also of its argument; this is called the Schwarz conjugate.
The above functions satisfy the unit determinant relations 
\begin{align}\label{det1relations}
  a \bar{a} - \lambda b \bar{b} = 1, \qquad A \bar{A} - \lambda B \bar{B} = 1, \qquad \mathcal{A} \bar{\mathcal{A}} - \lambda \mathcal{B} \bar{\mathcal{B}} = 1.
\end{align}
The functions $a$ and $b$ satisfy
\begin{align}\label{abasymptotics}
a(k) = 1 + O\bigg(\frac{1}{k}\bigg) + O\bigg(\frac{e^{2ikL}}{k}\bigg), \quad 
b(k) = O\bigg(\frac{1}{k}\bigg) + O\bigg(\frac{e^{2ikL}}{k}\bigg), \qquad k \to \infty, \ k \in \C,
\end{align}
whereas the functions $A$ and $B$ satisfy
\begin{align}\label{ABasymptotics}
A(k) = 1 + O\bigg(\frac{1}{k}\bigg) + O\bigg(\frac{e^{4ik^2T}}{k}\bigg), \quad 
B(k) = O\bigg(\frac{1}{k}\bigg) + O\bigg(\frac{e^{4ik^2T}}{k}\bigg), \qquad k \to \infty, \ k \in \C.
\end{align}
The eigenfunctions $\mu_j$ are related by
\begin{subequations}\label{murelations}
\begin{align}\label{mu3mu2}
&  \mu_3 = \mu_2 e^{-i\theta\hat{\sigma}_3} s(k), 
 	\\ \label{mu1mu2}
&  \mu_1 = \mu_2 e^{-i\theta\hat{\sigma}_3} S(k), 
 	\\ \label{mu4mu3}
 & \mu_4 = \mu_3 e^{-i(\theta - kL)\hat{\sigma}_3} S_L(k),
\end{align}
\end{subequations}
where 
\begin{align}\label{thetadef}
\theta = \theta(x,t,k) = kx+2k^2 t
\end{align}
and $e^{-i\theta \hat{\sigma}_3}s(k) = e^{-i\theta \sigma_3}s(k) e^{i\theta\sigma_3}$ etc. Using the entries of $s$, $S$, $S_L$, we construct the following quantities:
\begin{align*}
& \alpha(k)= a \mathcal{A} + \lambda \bar{b} \mathcal{B} e^{2ikL}, && 
 \beta(k)= b \mathcal{A} + \bar{a} \mathcal{B} e^{2ikL}, 
 	\\
& d(k)=a \bar{A}-\lambda b \bar{B}, && 
 \delta(k)=\alpha \bar{A}-\lambda \beta \bar{B}.
\end{align*}

\subsection{RH problem for $M$}
It was shown in \cite{FI2004} that the sectionally meromorphic function  $M(x,t,k)$  defined by
\begin{align}\label{Mdef}
M = \begin{cases} 
(\frac{[\mu_2]_1}{\alpha}, [\mu_4]_2), \quad & k \in D_1, \\
(\frac{[\mu_1]_1}{d}, [\mu_3]_2), & k \in D_2, \\
([\mu_3]_1, \frac{[\mu_1]_2}{\bar{d}}), & k \in D_3, \\
([\mu_4]_1, \frac{[\mu_2]_2}{\bar{\alpha}}), & k \in D_4,
\end{cases}
\end{align}
satisfies
\begin{subequations}\label{rhp}
\begin{align}\label{rhpa}
& M_-(x,t,k)=M_+(x,t,k)J(x,t,k),\qquad k \in \Sigma,
	\\
& M(x,t,k) = I + O(1/k), \qquad k \to \infty,
\end{align}
\end{subequations}
where $\Sigma = \R \cup i\R$ is oriented as in Figure \ref{jumps.pdf} and the jump matrix $J$ is defined by 
\begin{align}\label{J}
J=
\begin{cases}
J_2, & \arg k=0,\\
J_1, & \arg k=\pi/2,\\
J_4, & \arg k=\pi ,\\
J_3, & \arg k=3\pi/2,
\end{cases}
\end{align}
with
\begin{align}\nonumber
&J_1= \begin{pmatrix} \delta/d & -\mathcal{B}e^{2ikL}e^{-2i\theta} \\ \lambda \bar{B} e^{2i\theta}/(d\alpha) & a/\alpha \end{pmatrix}, && 
 J_2= \begin{pmatrix} 1 & -\beta e^{-2i\theta}/\bar{\alpha} \\ \lambda \bar{\beta} e^{2i\theta}/\alpha & 1/(\alpha \bar{\alpha}) \end{pmatrix},
	\\  \label{J1234def}
& J_3= \begin{pmatrix} \bar{\delta}/\bar{d} & -Be^{-2i\theta}/\bar{d}\bar{\alpha} \\ \lambda \bar{\mathcal{B}} e^{-2ikL} e^{2i\theta} & \bar{a}/\bar{\alpha} \end{pmatrix}, 
&&
 J_4=J_3 J_2^{-1} J_1.
\end{align}
Note that the jump matrix depends on $x$ and $t$ only via the function $\theta(x,t,k)$ defined in (\ref{thetadef}).  
The solution $q(x,t)$ of (\ref{NLS}) can be recovered from  $M$ via the identity
\begin{align}\label{sol}
q(x,t)=2i\lim_{k\rightarrow \infty} k M_{12}(x,t,k),
\end{align}
where the limit may be taken in any quadrant.

If the functions $\alpha(k)$ and $d(k)$ have no zeros, then the function $M$ is analytic for $k \in \C \setminus \Sigma$ and it can be characterized as the unique solution of the RH problem (\ref{rhp}) with jump matrix $J$. The jump matrix $J$  depends via the spectral functions on the initial datum  $q(x,0)$ as well as on the boundary values $q(0,t), q(L,t), q_x(0,t)$, and $q_x(L, t)$. If all these boundary values are known, then the value of $q(x,t)$ at any point $(x,t)$  can be obtained by solving the RH problem (\ref{rhp}) for $M$ and using (\ref{sol}).

If the functions $\alpha(k)$ and $d(k)$ have zeros, then $M$ may have pole singularities. The generic case of a finite number of simple poles can be treated by supplementing the RH problem with appropriate residue conditions, see \cite[Proposition 2.3]{FI2004}.


\section{The periodic problem}\label{periodicsec}
From now on, we restrict our attention to the periodic problem and assume that $q(x,t)$ satisfies
\begin{align}\label{qboundaryvalues}
q(0,t) = q(L,t) \quad \text{and} \quad q_x(0,t) = q_x(L,t) \quad \text{for} \quad  t \in [0, T].
\end{align}
In this case, we clearly have $\mathcal{A} = A$ and $\mathcal{B} = B$. Hence the jump matrices $J_i$ defined in (\ref{J1234def}) become
\begin{align}\nonumber
&J_1= \begin{pmatrix} \delta/d & -Be^{2ikL}e^{-2i\theta} \\ \lambda \bar{B} e^{2i\theta}/(d\alpha) & a/\alpha, \end{pmatrix} 
&&
J_2= \begin{pmatrix} 1 & -\beta e^{-2i\theta}/\bar{\alpha} \\ \lambda \bar{\beta} e^{2i\theta}/\alpha & 1/(\alpha \bar{\alpha}) \end{pmatrix},
	\\ \label{periodicJ1234def}
& J_3= \begin{pmatrix} \bar{\delta}/\bar{d} & -Be^{-2i\theta}/(\bar{d}\bar{\alpha}) \\ \lambda \bar{B} e^{-2ikL} e^{2i\theta} & \bar{a}/\bar{\alpha}
\end{pmatrix}, && 
J_4=J_3 J_2^{-1} J_1,
\end{align}
where
\begin{align}
\alpha = a A + \lambda \bar{b} B e^{2ik L}, \quad 
\beta = b A + \bar{a} B e^{2ikL}, \quad
d = a \bar{A} - \lambda b \bar{B}, \quad
\delta = \alpha\bar{A} - \lambda \beta \bar{B},
\end{align}
and the global relation (\ref{gr}) becomes
\begin{align}\label{periodicgr}
\lambda \bar{b} e^{2ikL} B^2 + (a - \bar{a} e^{2ikL})A B - b A^2 = e^{4ik^2 T} c^+.
\end{align}

Our goal is to find a representation for the solution $q(x,t)$ in terms of the initial datum $q_0(x) = q(x,0)$. Since the expression (\ref{periodicJ1234def}) for the jump matrix depends via the spectral functions $A(k)$  and $B(k)$ on the two (unknown) functions $q(0,t)$ and $q_x(0,t)$ given in (\ref{qboundaryvalues}), the representation (\ref{sol}) does not achieve this goal. 
However, in what follows we will show that it is possible to eliminate $A$  and $B$ from the formulation of the RH problem by performing two steps. In the first step, which is presented in this section, we will transform the RH problem (\ref{rhp}) to a new RH problem which depends on the unknown boundary values in (\ref{qboundaryvalues}) only via the quotient 
\begin{align}\label{Gammadef}
\Gamma(k) = \frac{B(k)}{A(k)}.
\end{align}
In the second step, which is presented in Section \ref{mainsec}, we will show that $\Gamma(k)$ can be effectively replaced by a function $\tilde{\Gamma}(k)$, which only depends on the initial datum. 


\subsection{RH problem for $m$}
Consider the sectionally meromorphic function $m(x,t,k)$ defined in terms of the eigenfunctions $\mu_j(x,t,k)$, $j = 1,\dots, 4$, by
\begin{align}\label{mdef}
m = \begin{cases} 
(\frac{A[\mu_2]_1}{\alpha}, \frac{[\mu_4]_2}{A}), \quad & k \in D_1, \\
(\frac{[\mu_1]_1}{d}, [\mu_3]_2), & k \in D_2, \\
([\mu_3]_1, \frac{[\mu_1]_2}{\bar{d}}), & k \in D_3, \\
(\frac{[\mu_4]_1}{\bar{A}}, \frac{\bar{A} [\mu_2]_2}{\bar{\alpha}}), & k \in D_4.
\end{cases}
\end{align}
The function $m$  is related to the solution $M$ of (\ref{rhp}) by 
$$m = MH,$$
where the sectionally meromorphic function $H$ is defined by
$$H_1 = \begin{pmatrix} A & 0 \\ 0 & 1/A \end{pmatrix}, \qquad
H_2 = H_3 = I,\qquad
H_4 = \begin{pmatrix} 1/\bar{A} & 0 \\ 0 & \bar{A} \end{pmatrix};$$
here $H_j$ denotes the restriction of $H$ to $D_j$, $j = 1,2,3,4$. 
The function $m$ may have poles at the possible zeros of the entire functions $\alpha$, $A$, and $d$. In order to express the locations of these possible poles and the associated residue conditions in terms of only $a,b, \Gamma$, we introduce the functions $\eta(k)$ and $\xi(k)$ which are defined by
\begin{align}\label{etaxidef}
\eta(k) = 
- \frac{ \lambda \bar{\beta} A^2}{\alpha}
= \frac{ \lambda (\bar{b} + a \bar{\Gamma} e^{-2ikL})}{(a + \lambda \bar{b} \Gamma e^{2ikL})(\lambda \Gamma \bar{\Gamma} - 1)}, \qquad 
\xi(k) = \frac{\lambda (\bar{a}\bar{B} - \bar{b} \bar{A})}{d}
= \frac{\lambda(\bar{a} \bar{\Gamma} - \bar{b})}{a  - \lambda b\bar{\Gamma}}.
\end{align}
Clearly, $\Gamma$, $\eta$, and $\xi$ are meromorphic functions of $k \in \C$. We will consider the generic situation in which the possible poles of these functions satisfy the following assumption.

\begin{assumption}\label{poleassumption}
We assume the following:
\begin{enumerate}[$-$]
\item $\Gamma$, $\eta$, and $\xi$ have no poles on the contour $\Sigma = \R \cup i\R$.

\item In $D_1$, $\Gamma(k)$ has at most finitely many poles $\{k_j\}_1^{N_1} \subset D_1$ and these poles are all simple.

\item In $D_1$, $\eta(k)$ has at most finitely many poles $\{K_j\}_1^{N_2} \subset D_1$ and these poles are all simple and disjoint from the poles of $\Gamma$.

\item In $D_2$, $\xi(k)$ has at most finitely many poles $\{\kappa_j\}_1^{N_3} \subset D_2$ and these poles are all simple.

\item Let $P = \{k_j, \bar{k}_j\}_1^{N_1} \cup \{K_j, \bar{K}_j\}_1^{N_2} \cup \{\kappa_j, \bar{\kappa}_j\}_1^{N_3}$ denote the set of poles and their complex conjugates. Then $P$ is disjoint from the set of zeros of $a$ and $b$.
\end{enumerate}
\end{assumption}

\begin{remark}
Regarding Assumption \ref{poleassumption} it is noted that it is {\it not} important to investigate whether these poles exist, since, remarkably, they cancel out in the formulation of the final RH problem. Actually, as shown in the main theorem (Theorem \ref{mainth}), we only need to worry about the poles of $\tilde{\Gamma}$ and the existence of these poles is discussed in Remark \ref{polesremark}.
\end{remark}

We will show that $m$ satisfies the following RH problem.

\begin{RHproblem}{\bf(The RH problem for $m$)}\label{RHm}
Find a $2 \times 2$-matrix valued function $m(x,t,k)$ with the following properties:
\begin{itemize}
\item $m(x,t,\cdot) : \C \setminus (\Sigma \cup P) \to \C^{2 \times 2}$ is analytic.

\item The limits of $m(x,t,k)$ as $k$ approaches $\Sigma \setminus \{0\}$ from the left and right exist, are continuous on $\Sigma \setminus \{0\}$, and satisfy
\begin{align}\label{mjump}
  m_-(x,t,k) = m_+(x, t, k) v(x, t, k), \qquad k \in \Sigma \setminus \{0\},
\end{align}
where the jump matrix $v$ is defined by
\begin{align}\nonumber
&  v_1 = \begin{pmatrix}  \frac{a  - \lambda b \bar{\Gamma} - \lambda \Gamma (\bar{a}\bar{\Gamma} - \bar{b})e^{2ikL} }{a- \lambda b \bar{\Gamma}} & -\Gamma  e^{2ikL} e^{-2 i \theta } \\
 \frac{\lambda \bar{\Gamma} e^{2 i \theta } }{(a - \lambda b \bar{\Gamma}) (a + \lambda\bar{b} \Gamma  e^{2ikL})} & \frac{a}{a + \lambda \bar{b} \Gamma  e^{2ikL}} \end{pmatrix}, && \arg k = \frac{\pi}{2},
	\\ \nonumber
& v_2 =  \begin{pmatrix} 1 - \lambda \Gamma \bar{\Gamma} 
& -\frac{(\bar{a} \Gamma e^{2ikL}+b) e^{-2 i \theta }}{\bar{a}+ \lambda b \bar{\Gamma} e^{-2ikL}} \\
 \frac{\lambda (a \bar{\Gamma}e^{-2ikL}+\bar{b}) e^{2 i \theta } }{a+ \lambda \bar{b} \Gamma  e^{2ikL} } &
   \frac{1}{(a+ \lambda \bar{b} \Gamma  e^{2ikL}) (\bar{a} + \lambda b \bar{\Gamma} e^{-2ikL})}  \end{pmatrix},&& \arg k = 0,
   	\\  \nonumber
& v_3 = \begin{pmatrix} \frac{\bar{a}  - \lambda \bar{b} \Gamma - \lambda \bar{\Gamma}(a\Gamma - b) e^{-2ikL}}{\bar{a} - \lambda \bar{b} \Gamma } & - \frac{\Gamma e^{-2 i \theta}}{(\bar{a} - \lambda \bar{b} \Gamma) (\bar{a} + \lambda b \bar{\Gamma} e^{-2ikL} )} \\
 \lambda \bar{\Gamma} e^{-2ikL} e^{2 i \theta} & \frac{\bar{a}}{\bar{a} + \lambda b \bar{\Gamma} e^{-2ikL}}  \end{pmatrix}, && \arg k = -\frac{\pi}{2},
   	\\ \label{vdef}
& v_4 = \begin{pmatrix}	
 \frac{1 - \lambda \Gamma \bar{\Gamma}}{(a - \lambda b \bar{\Gamma}) (\bar{a} - \lambda \bar{b} \Gamma )} & -\frac{(a \Gamma - b)e^{-2 i \theta } }{\bar{a} - \lambda \bar{b} \Gamma } \\
 \frac{\lambda(\bar{a} \bar{\Gamma} - \bar{b})e^{2 i \theta }}{a - \lambda b \bar{\Gamma}} & 1
 \end{pmatrix}, && \arg k = \pi.
\end{align}

\item $m(x,t,k) = I + O(k^{-1})$ as $k \to \infty$.

\item $m(x,t,k) = O(1)$ as $k \to 0$.

\item At the points $k_j \in D_1$  and  $\bar{k}_j \in D_4$, $m$ satisfies, for $j = 1, \dots, N_1$,
\begin{subequations}\label{Gammaresidues}
\begin{align}\label{Gammaresiduesa}
[m(x,t,k)]_1 = O(k-k_j), \quad
[m(x,t,k)]_2 = O\bigg(\frac{1}{k-k_j}\bigg), \qquad k \to k_j,
	\\\label{Gammaresiduesb}
[m(x,t,k)]_1 = O\bigg(\frac{1}{k-\bar{k}_j}\bigg), \quad
[m(x,t,k)]_2 = O(k-\bar{k}_j),  \qquad k \to \bar{k}_j.
\end{align}
\end{subequations}

\item At the points $K_j \in D_1$  and  $\bar{K}_j \in D_4$, $m$ has at most simple poles and the residues at these poles satisfy, for $j = 1, \dots, N_2$,
\begin{subequations}\label{etaresidues}
\begin{align}\label{etaresiduesa}
& \underset{k = K_j}{\res} [m(x,t,k)]_1 =  [m(x,t,K_j)]_2 e^{2i\theta(x,t,K_j)} \underset{k=K_j}{\res} \eta(k),
	\\\label{etaresiduesb}
& \underset{k=\bar{K}_j}{\res} [m(x,t,k)]_2 =  \lambda [m(x,t,\bar{K}_j)]_1 e^{-2i\theta(x,t,\bar{K}_j)} \overline{\underset{k=K_j}{\res} \eta(k)}.	
\end{align}
\end{subequations}

\item At the points $\kappa_j \in D_2$  and  $\bar{\kappa}_j \in D_3$, $m$ has at most simple poles and the residues at these poles satisfy, for $j = 1, \dots, N_3$,
\begin{subequations}\label{xiresidues}
\begin{align}\label{xiresiduesa}
& \underset{k=\kappa_j}{\res} [m(x,t,k)]_1 =  [m(x,t,\kappa_j)]_2 e^{2i\theta(x,t,\kappa_j)} \underset{k=\kappa_j}{\res} \xi(k),
	\\\label{xiresiduesb}
& \underset{k=\bar{\kappa}_j}{\res} [m(x,t,k)]_2 = \lambda [m(x,t,\bar{\kappa}_j)]_1 e^{-2i\theta(x,t,\bar{\kappa}_j)} \overline{\underset{k=\kappa_j}{\res} \xi(k)}.
\end{align}
\end{subequations}

\end{itemize}
\end{RHproblem}

Note that the above RH problem depends on the functions $A$ and $B$ only via their quotient $\Gamma = B/A$. The next proposition shows that the RH problem for $m$ can be used to determine the solution $q$ of (\ref{NLS}) on the circle of length $L$, assuming that the initial datum and the quotient $B/A$ are known.

\begin{proposition}\label{mprop}
Let $0 < T < \infty$. Suppose $q(x,t)$ is a smooth solution of (\ref{NLS}) for $(x,t) \in \R \times [0,T]$ which is $x$-periodic of period $L > 0$, i.e., $q(x+L, t) = q(x,t)$.
Define $a,b, \Gamma$ by (\ref{abdef}) and (\ref{Gammadef}) and let $\eta(k)$  and $\xi(k)$ be given by (\ref{etaxidef}). 
Suppose Assumption \ref{poleassumption} holds. 
Then the RH problem \ref{RHm} has a unique solution $m(x,t,k)$ for each $(x,t) \in [0,L] \times [0,T]$. The solution $q$ can be obtained from  $m$ via the relation
\begin{align}\label{recoverqm}
  q(x,t) = 2i\lim_{k \to \infty} k m_{12}(x,t,k), \qquad (x,t) \in [0,L] \times [0,T].
\end{align}
Moreover, $\det m = 1$ and $m$ obeys the symmetries
\begin{align}\label{msymm}
  m_{11}(x,t,k) = \overline{m_{22}(x,t,\bar{k})}, \qquad 
  m_{21}(x,t,k) = \lambda \overline{m_{12}(x,t,\bar{k})}.
\end{align}
\end{proposition}
\begin{proof}
Let us first show that the solution of the RH problem \ref{RHm} is unique. 
If the set $P$ of poles is empty, then uniqueness follows by standard considerations because the jump matrix has unit determinant. 
The problem with a nonempty set $P$ can be transformed into a problem for which $P$ is empty. Indeed, if the set $\{k_j\}_1^{N_1}$ of poles of $\Gamma$ is nonempty, then the function
\begin{align}
m \begin{pmatrix} \prod_{j=1}^{N_1} \frac{k- \bar{k}_j}{k-k_j} & 0 \\ 0 & \prod_{j=1}^{N_1} \frac{k-k_j}{k- \bar{k}_j} \end{pmatrix}
\end{align}
satisfies an analogous RH problem but with no singularities at the points $\{k_j, \bar{k}_j\}_1^{N_1}$.
The residue conditions (\ref{etaresidues}) and (\ref{xiresidues}) are of a standard form, so the possible poles $\{K_j\}_1^{N_2}$ and $\{\kappa_j\}_1^{N_3}$ of $\eta$ and $\xi$  can  be regularized in the standard way, see e.g. \cite{FI1996}. This proves uniqueness. It also follows from these arguments that $\det m = 1$. 

The symmetries (\ref{msymm}) can be expressed as $m(x,t,k) = \sigma_1 \overline{m(x,t,\bar{k})} \sigma_1$ if $\lambda = 1$ and as $m(x,t,k) = \sigma_3 \sigma_1 \overline{m(x,t,\bar{k})}\sigma_1 \sigma_3$ if $\lambda = -1$.
These symmetries follow from the uniqueness of the solution and the fact that $v$ obeys the symmetries
\begin{align}\label{vsymm}
v(x,t,k) = \begin{cases}
\sigma_1 \overline{v(x,t,\bar{k})}^{-1} \sigma_1, & \lambda = 1,\\
\sigma_3 \sigma_1 \overline{v(x,t,\bar{k})}^{-1} \sigma_1 \sigma_3, & \lambda = -1.
 \end{cases}
 \end{align}

Define $m$ by (\ref{mdef}). 
Long but straightforward computations using the unit determinant relations (\ref{det1relations}) show that the jump matrix $v$ given in (\ref{vdef}) satisfies
$$v_1 = H_1^{-1} J_1 H_2, \quad
v_2 = H_1^{-1} J_2 H_4, \quad
v_3 = H_3^{-1} J_3 H_4, \quad
v_4 = H_3^{-1} J_4 H_2.$$
Since $m = MH$, the jump relation (\ref{mjump}) follows from (\ref{rhpa}).
Alternatively, (\ref{mjump}) can be derived directly from (\ref{murelations}) and (\ref{mdef}). 

We next show that $m$ is analytic for $k \in \C \setminus (\Sigma \cup P)$ and establish the residue conditions (\ref{Gammaresidues})-(\ref{xiresidues}). We consider the columns of $m$ in each quadrant separately.

{\bf The second column of $m$ in $D_1$.}
The condition $A \bar{A} - \lambda B \bar{B} = 1$ implies that $A$ and  $B$ cannot have any common zeros. Hence, the set of zeros of $A$ coincides with the set of poles of $\Gamma$, and at each such pole $B$ is nonzero.  
Since $[m]_2 = [\mu_4]_2/A$, it follows that $[m]_2$  is analytic in $D_1$ except at the possible poles $k_j$ of $\Gamma$.

Equations (\ref{mu3mu2}) and (\ref{mu4mu3}) imply that
\begin{align}\label{mu4mu2}
  \mu_4 = \mu_2 \begin{pmatrix} \bar{\alpha} & \beta e^{-2i\theta} \\ 
  \lambda \bar{\beta} e^{2i\theta} & \alpha \end{pmatrix}.
\end{align}
The second column of (\ref{mu4mu2}) is
\begin{align}\label{mu42mu21}
[\mu_4]_2 = [\mu_2]_1 \beta e^{-2i\theta} + [\mu_2]_2 \alpha.
\end{align}
Dividing by $A$, this becomes, for $k \in D_1$,
\begin{align}\nonumber
[m]_2 & = [\mu_2]_1 \frac{\beta}{A} e^{-2i\theta} + [\mu_2]_2 \frac{\alpha}{A}
= [m]_1 \frac{\alpha \beta}{A^2} e^{-2i\theta} + [\mu_2]_2 \frac{\alpha}{A}
	\\ \label{m2inD1}
& = [m]_1 (b + \bar{a}\Gamma e^{2ikL})(a + \lambda \bar{b} \Gamma e^{2ikL}) e^{-2i\theta} + [\mu_2]_2 (a + \lambda \bar{b} \Gamma e^{2ikL}).
\end{align}
Since $\bar{b}$ is nonzero at each point of $P$ by assumption and since $A(k_j) = 0$ and $B(k_j) \neq 0$, it follows that $\alpha(k_j) = \lambda \overline{b(\bar{k}_j)} B(k_j) e^{2ik_jL} \neq 0$. 
Thus $[m]_1 = A[\mu_2]_1/\alpha$ has at least a simple zero at $k_j$ and it then follows from (\ref{m2inD1}) that $[m]_2$ has (at most) a simple pole at $k_j$. 
This proves (\ref{Gammaresiduesa}).

{\bf The first column of $m$ in $D_1$.}
Equation (\ref{mu4mu2}) can be written as 
\begin{align}\label{mu2mu4}
  \mu_2 = \mu_4 \begin{pmatrix} \alpha & -\beta e^{-2i\theta} \\ 
  -\lambda \bar{\beta} e^{2i\theta} & \bar{\alpha} \end{pmatrix},
\end{align}
and the first column of (\ref{mu2mu4}) yields
$[\mu_2]_1 = [\mu_4]_1 \alpha - [\mu_4]_2  \lambda \bar{\beta} e^{2i\theta}$.
Thus, for $k \in D_1$, 
\begin{align}\label{m1inD1}
[m]_1 = [\mu_4]_1 A - [\mu_4]_2 \frac{ \lambda \bar{\beta} A e^{2i\theta}}{\alpha}
= [\mu_4]_1 A - [m]_2 \frac{ \lambda \bar{\beta} A^2 e^{2i\theta}}{\alpha}
= [\mu_4]_1 A + [m]_2 \eta e^{2i\theta}.
\end{align}
We already saw that $[m]_1$ vanishes at each possible pole $k_j$ of $[m]_2$. 
It follows that $[m]_1$ is analytic in $D_1$ except at the possible simple poles $K_j$ of $\eta$, and at each such pole the residue condition (\ref{etaresiduesa}) holds.

{\bf The second column of $m$ in $D_2$.}
This column is analytic in $D_2$, because $\mu_3$  is entire.

{\bf The first column of $m$ in $D_2$.}
Equations (\ref{mu3mu2}) and (\ref{mu1mu2}) imply 
\begin{align}\label{mu1mu3}
  \mu_1 = \mu_3 e^{-i\theta\hat{\sigma}_3} (s^{-1} S), \quad \text{where} \quad
  s^{-1} S = \begin{pmatrix}
d & aB - bA \\ \lambda(\bar{a}\bar{B} - \bar{b} \bar{A}) & \bar{d} 
\end{pmatrix}.
\end{align}
The first column of (\ref{mu1mu3}) can be written as
\begin{align}\label{mu11mu31}
[\mu_1]_1 = [\mu_3]_1 d + [\mu_3]_2 \lambda(\bar{a}\bar{B} - \bar{b} \bar{A}) e^{2i\theta}.\end{align}
Dividing by $d$, this becomes, for $k \in D_2$,
\begin{align}\label{m1inD2}
[m]_1 = [\mu_3]_1 + [\mu_3]_2 \frac{\lambda(\bar{a}\bar{B} - \bar{b} \bar{A})}{d} e^{2i\theta}
= [\mu_3]_1 + [m]_2 \xi e^{2i\theta}.
\end{align}
It follows that $[m]_1$ is analytic in $D_2$ except at the simple poles $\kappa_j$ of $\xi$, and at each such pole the residue condition (\ref{xiresiduesa}) holds.

This completes the proof of the analyticity of $m$ and the residue conditions in the upper half-plane.
The analogous properties in the lower half-plane follow from similar arguments or from the symmetries (\ref{msymm}). 
Together with Assumption \ref{poleassumption}, equations (\ref{m2inD1}), (\ref{m1inD1}), and (\ref{m1inD2}) also imply that $m$ has continuous boundary values on $\Sigma \setminus \{0\}$ and that $m = O(1)$ as $k \to 0$. 
The behavior of $m$ as $k \to \infty$  is a consequence of (\ref{mdef}). Finally, (\ref{recoverqm}) follows from (\ref{sol}).
\end{proof}

\section{Main result}\label{mainsec}
The RH problem \ref{RHm} for $m$ depends on the final time $T$ via the function  $\Gamma = B/A$. However, the solution $q(x,t)$ is independent of $T$. This suggests that it should be possible to eliminate the $T$-dependence from the RH problem \ref{RHm}.
In this section, we define a new $T$-independent RH problem---henceforth called the RH problem for $\tilde{m}$---by applying an appropriate deformation to the RH problem for $m$. The basic idea of this deformation is to replace $\Gamma$ by a new $T$-independent function  $\tilde{\Gamma}$. Since $\tilde{\Gamma}$ is defined in terms of the spectral functions $a(k)$ and $b(k)$ alone, 
this will lead to our main result.

\subsection{Motivating remarks}
As motivation for the definition of the RH problem for $\tilde{m}$, we temporarily consider a solution $q(x,t)$ of the NLS equation on the interval $[0,L]$ whose boundary values have decay as $t \to \infty$. In this case, it can be shown that the functions $A$, $B$, $\mathcal{A}$, and $\mathcal{B}$ have finite limits as $T \to \infty$; we denote these limits by $\tilde{A}$, $\tilde{B}$, $\tilde{\mathcal{A}}$, and $\tilde{\mathcal{B}}$:
$$\tilde{A} = \lim_{T\to \infty} A, \qquad \tilde{B} = \lim_{T\to \infty} B, \qquad \tilde{\mathcal{A}} = \lim_{T\to \infty} \mathcal{A}, \qquad \tilde{\mathcal{B}} = \lim_{T\to \infty} \mathcal{B}.$$
Taking the same limit in the global relation (\ref{gr}) leads to the equation
\begin{align}\label{tildegr}
(a\tilde{\mathcal{A}} + \lambda \bar{b} e^{2ikL}\tilde{\mathcal{B}})\tilde{B} - (b \tilde{\mathcal{A}} + \bar{a} e^{2ikL}\tilde{\mathcal{B}})\tilde{A} = 0, \qquad k \in D_1 \cup D_3,
\end{align}
which can be viewed as a relation between the two quotients $\tilde{B}/\tilde{A}$ and $\tilde{\mathcal{B}}/\tilde{\mathcal{A}}$. In the spatially periodic setting, these two quotients are equal and hence (\ref{tildegr}) can be solved for $\tilde{B}/\tilde{A}$ with the result that
\begin{align}\label{tildeGammamotivation}
\frac{\tilde{B}}{\tilde{A}} = \frac{\lambda}{2 e^{ikL} \bar{b}} \Big(\bar{a} e^{ikL} - a e^{-ikL} \pm i\sqrt{4 - \Delta^2}\Big),
\end{align}
where $\Delta = a e^{-ikL}+\bar{a} e^{ikL}$.
This suggests that we define the new RH problem for $\tilde{m}$ by replacing $\Gamma$ with the $T$-independent function $\tilde{\Gamma}$, where $\tilde{\Gamma} \equiv \tilde{B}/\tilde{A}$ is defined by the right-hand side of (\ref{tildeGammamotivation}). Since the right-hand side of (\ref{tildeGammamotivation}) is defined in terms of the initial datum alone, this leads to an effective solution of the $x$-periodic initial value problem. 

We first need to give a careful definition of $\tilde{\Gamma}$ in which the branch of the square root in (\ref{tildeGammamotivation}) is fully specified.

\subsection{Definition of $\tilde{\Gamma}$}
We define the function $\tilde{\Gamma}(k)$ by 
\begin{align}\label{tildeGammadef}
\tilde{\Gamma} = \frac{\lambda}{2 e^{ikL} \bar{b}} \Big(\bar{a} e^{ikL} - a e^{-ikL} - i\sqrt{4 - \Delta^2}\Big),
\end{align}
where $\Delta(k)$ is given by 
\begin{align}\label{Deltadef}
  \Delta = a e^{-ikL}+\bar{a} e^{ikL}, \qquad k \in \C.
\end{align}

\begin{remark}
The function $\tilde{\Gamma}(k)$ defined in (\ref{tildeGammadef}) is related to the classical Titchmarsh--Weyl function appearing in the spectral theory of the Hill operator.
\end{remark}

In order to make the definition (\ref{tildeGammadef}) of $\tilde{\Gamma}$ precise, we need to introduce branch cuts and fix the sign of $\sqrt{4 - \Delta^2}$. 
To this end, consider the zero-set $\mathcal{P}$ of the entire function $4 - \Delta^2$:
$$\mathcal{P} = \{k \in \C \, | \, 4 - \Delta(k)^2 = 0 \}.$$
The function $\Delta(k)$ is the trace of the so-called monodromy matrix and features heavily in the classical approach to the $x$-periodic problem. In that framework, $\mathcal{P}$ is the periodic spectrum\footnote{The periodic spectrum is defined as the union of all periodic and antiperiodic eigenvalues. 
} and is well studied. In what follows, we collect some well-known facts about $\mathcal{P}$, see e.g. \cite{GK2014}, and define a set of branch cuts $\mathcal{C}$ such that $\sqrt{4 - \Delta^2}$ becomes single-valued on $\C \setminus \mathcal{C}$.

In the case of vanishing initial datum $q_0 \equiv 0$, we have $a = 1$, thus $\Delta(k) = 2\cos(kL)$, and hence $4 - \Delta^2 = 4\sin^2(kL)$ has a double zero at each of the points $n\pi/L$, $n \in \Z$. According to the so-called Counting Lemma (see \cite[Lemma 6.3]{GK2014}), the zero-set of $4 - \Delta^2$ has a similar structure for large $k$ for any initial datum $q_0 \in L^2([0,L])$ in the following sense: For $n \in \Z$, let 
\begin{align}\label{Dndef}
\mathcal{D}_n = \bigg\{k \in \C \, \bigg| \, \bigg|k - \frac{n\pi}{L}\bigg| < \frac{\pi}{4L}\bigg\}
\end{align}
denote the open disk of radius $\pi/(4L)$ centered at $n\pi/L$. Then, there is an integer $N > 0$ such that, counted with multiplicity, $4 - \Delta^2$ has exactly two roots in each disk $\mathcal{D}_n$ with $|n| > N$ and exactly $4N + 2$ roots in the disk $\{|k| < N\pi/L + \pi/(4L)\}$.
Using this result and employing the lexicographic ordering of complex numbers
$$u \preccurlyeq v \quad \text{if and only if} \quad \begin{cases} \re u < \re v \\
\qquad \text{or} \\
\text{$\re u = \re v$ and $\im u \leq \im v$},\end{cases}$$
we can write $\mathcal{P}$ as
\begin{align}\label{calPdef}
  \mathcal{P} = \bigcup_{n\in \Z} \{\lambda_n^+, \lambda_n^-\}
\end{align}
where 
$$\cdots \preccurlyeq \lambda_{n-1}^+ \preccurlyeq \lambda_n^- \preccurlyeq \lambda_n^+ \preccurlyeq \lambda_{n+1}^- \preccurlyeq \cdots, \qquad n \in \Z,$$
and $\lambda_n^\pm$ belong to $\mathcal{D}_n$ for all large enough  $|n|$.
The symmetry $\Delta(k) = \overline{\Delta(\bar{k})}$ implies that the periodic spectrum $\mathcal{P}$ is invariant under complex conjugation and that $\Delta$ is real-valued on $\R$. To define $\mathcal{C}$, we consider the two cases $\lambda = 1$ and $\lambda = -1$ separately.

In the defocusing case (i.e., $\lambda = 1$), the periodic spectrum $\mathcal{P}$ is purely real and 
$$\lambda_{n-1}^+ < \lambda_n^- \leq \lambda_n^+ < \lambda_{n+1}^-, \qquad n \in \Z.$$
The open interval $(\lambda_n^-, \lambda_n^+)$ is called the $n$th spectral gap whenever it is nonempty. 
For $\lambda = 1$, we define $\mathcal{C}$ as the union of all spectral gaps:
\begin{align}\label{calCdef1}
\mathcal{C} = \bigcup_{n \in \Z}  (\lambda_n^-, \lambda_n^+), \qquad \lambda = 1.
\end{align}
Thus, in this case $\mathcal{C}$ is a union of open subintervals of $\R$. The $n$th interval $(\lambda_n^-, \lambda_n^+)$ is contained in $\mathcal{D}_n$ for all sufficiently large $|n|$.

In the focusing case (i.e., $\lambda = -1$), $\mathcal{P}$ is typically not a subset of $\R$. If $z_1, z_2 \in \C$, we let $(z_1, z_2)$ denote the open straight-line segment from $z_1$ to $z_2$, i.e.,
$$(z_1, z_2) = \{z_1 + t(z_2-z_1) \in \C \, | \, 0 < t < 1 \}.$$
Using this notation, it is possible to define $\mathcal{C}$ by (\ref{calCdef1}), but this definition has the disadvantage that it may break the existing symmetry under complex conjugation. Therefore, we instead define $\mathcal{C}$ as follows.
Let $N >0$ be as in the Counting Lemma so that there are $4N + 2$ roots counted with multiplicity in the disk $\{|k| < N\pi/L + \pi/(4L)\}$. An even number, say $2M$, of these $4N + 2$ roots have odd multiplicity; let $\{\lambda_j^{\text{odd}}\}_1^{2M}$ denote these roots. 
Since $\Delta$  is real-valued on $\R$ and $\Delta = 2 \cos(kL) + o(1)$ as $k \to \pm \infty$, an even number, say $2\hat{M}$, of the roots $\{\lambda_j^{\text{odd}}\}_1^{2M}$ are real; let $\{\hat{\lambda}_j^-, \hat{\lambda}_j^+\}_{j=1}^{\hat{M}}$ denote these real roots ordered so that
$$\hat{\lambda}_1^- \leq \hat{\lambda}_1^+ \leq \cdots \leq \hat{\lambda}_{\hat{M}}^- \leq \hat{\lambda}_{\hat{M}}^+.$$
Let $\{\tilde{\lambda}_j^-, \tilde{\lambda}_j^+\}_{j=1}^{M-\hat{M}}$ denote the remaining odd-order roots in the disk $\{|k| < N\pi/L + \pi/(4L)\}$ ordered lexicographically:
$$\tilde{\lambda}_1^- \preccurlyeq \tilde{\lambda}_1^+ \preccurlyeq \cdots \preccurlyeq \lambda_{M-\hat{M}}^- \preccurlyeq \lambda_{M-\hat{M}}^+.$$
Note that there is an even number of roots $\tilde{\lambda}_j^\pm$ on any vertical line  $\re k = \text{constant}$.
We define $\mathcal{C}$ as the union of the open real intervals $(\hat{\lambda}_j^-, \hat{\lambda}_j^+)$, the vertical line segments $(\tilde{\lambda}_j^-, \tilde{\lambda}_j^+)$, as well as the open real intervals $(\lambda_n^-, \lambda_n^+)$, $|n|>N$:
\begin{align}\label{calCdef2}
\mathcal{C} = \bigcup_{j=1}^{\hat{M}} (\hat{\lambda}_j^-, \hat{\lambda}_j^+) \cup \bigcup_{j=1}^{M-\hat{M}} (\tilde{\lambda}_j^-, \tilde{\lambda}_j^+) \cup \bigcup_{|n| > N} (\lambda_n^-, \lambda_n^+), \qquad \lambda = -1.
\end{align}

In both the defocusing and the focusing case, we orient the contour $\mathcal{C}$ so that (i) any part of $\mathcal{C}$ that is contained in $\Sigma = \R \cup i\R$ is oriented in the same direction as $\Sigma$, i.e., $\mathcal{C} \cap i\R_+$, $\mathcal{C} \cap \R_+$, $\mathcal{C} \cap i\R_-$, and $\mathcal{C} \cap \R_-$ are oriented down, right, up, and left, respectively, and (ii) the subcontours $\mathcal{C} \cap (D_1 \cup D_2)$ and $\mathcal{C} \cap (D_3 \cup D_4)$ are oriented downward and upward, respectively. See Figure \ref{contourfig} for two examples of $\mathcal{C}$.

We can now complete the definition of $\tilde{\Gamma}$. 
The function $\sqrt{4 - \Delta(k)^2}$ in (\ref{tildeGammadef}) is single-valued for $k \in \C \setminus \mathcal{C}$ up to a choice of sign. 
This sign can be fixed by considering the large $k$ behavior of $\tilde{\Gamma}$.
Indeed, as $k \to \infty$ in $\R$, (\ref{abasymptotics}) implies that $a \sim 1$ and  $\bar{a} \sim 1$, and so 
$$\Delta = a e^{-ikL}+\bar{a} e^{ikL} \sim 2\cos(kL), \qquad \sqrt{4 - \Delta^2} \sim \pm 2\sin(kL).$$ 
Since we would like to have $\tilde{\Gamma} = O(1/k)$ as $k \to \infty$ in $\R$ (at least if $k$ stays away from the disks $\mathcal{D}_n$), we fix the branch of $\sqrt{4 - \Delta^2}$ by requiring that (see (\ref{sqrt4Deltaasymptotics}) for a more detailed estimate)
\begin{align}\label{sqrt4Delta2}
\sqrt{4 - \Delta^2} \sim 2\sin(kL), \qquad k \to \infty, \ k \in \R  \setminus \cup_{n\in \Z} \mathcal{D}_n.
\end{align}
In summary, $\tilde{\Gamma}: \C \setminus \mathcal{C} \to \C$ is defined by (\ref{tildeGammadef}) with the branch of $\sqrt{4 - \Delta^2}$ fixed by (\ref{sqrt4Delta2}).

\subsection{RH problem for $\tilde{m}$}
Define the function $g(k)$ by 
\begin{align}\nonumber
& g_1 = \begin{pmatrix} \frac{a + \lambda \bar{b} \Gamma e^{2ikL}}{a + \lambda \bar{b} \tilde{\Gamma} e^{2ikL}} & (\tilde{\Gamma} - \Gamma) e^{2ikL} e^{-2i\theta} \\ 0 & \frac{a + \lambda \bar{b} \tilde{\Gamma} e^{2ikL}}{a + \lambda \bar{b}\Gamma e^{2ikL}} \end{pmatrix}, \qquad
g_2 = \begin{pmatrix} 1 & 0 \\ 
\frac{\lambda (\bar{\tilde{\Gamma}} - \bar{\Gamma})e^{2i\theta}}{(a - \lambda b \bar{\Gamma})(a - \lambda b \bar{\tilde{\Gamma}})} & 1 \end{pmatrix},
	\\ \label{gdef}
& g_3 = \begin{pmatrix} 1 & \frac{(\tilde{\Gamma} - \Gamma)e^{-2i\theta}}{(\bar{a} - \lambda \bar{b} \Gamma)(\bar{a} - \lambda \bar{b} \tilde{\Gamma})} \\
0 & 1 \end{pmatrix}, \qquad
g_4 = \begin{pmatrix} \frac{\bar{a} + \lambda b \bar{\tilde{\Gamma}} e^{-2ikL}}{\bar{a} + \lambda b \bar{\Gamma} e^{-2ikL}} & 0 \\
\lambda(\bar{\tilde{\Gamma}} - \bar{\Gamma}) e^{-2ikL} e^{2i\theta} & \frac{\bar{a} + \lambda b \bar{\Gamma} e^{-2ikL}}{\bar{a} + \lambda b \bar{\tilde{\Gamma}} e^{-2ikL}} \end{pmatrix},
\end{align} 
where $g_j$ denotes the restriction of $g$ to  $D_j$ for $j = 1, \dots, 4$.
We introduce $\tilde{m}(x,t,k)$ by
\begin{align}\label{tildemdef}
\tilde{m} = mg,
\end{align}
where $m$ is the solution of RH problem \ref{RHm}. 
The function $g$ is defined in such a way that $\tilde{m}$ satisfies the jump relation $\tilde{m}_- = \tilde{m}_+ \tilde{v}$ on $\Sigma \setminus \mathcal{C}$, where the jump matrix $\tilde{v}$ is given by the same expression (\ref{vdef}) as $v$ except that $\Gamma$  is replaced by $\tilde{\Gamma}$.
This follows by a direct computation using that 
\begin{align}\label{tildevgv}
\tilde{v}_1 = g_1^{-1} v_1 g_2, \qquad
\tilde{v}_2 = g_1^{-1} v_2 g_4, \qquad
\tilde{v}_3 = g_3^{-1} v_3 g_4, \qquad
\tilde{v}_4 = g_3^{-1} v_4 g_2.
\end{align}
However, because of the square root $\sqrt{4 - \Delta^2}$, $\tilde{v}$ is not given by the same expression as $v$ on $\Sigma \cap \mathcal{C}$, and if $\lambda = -1$, then $\tilde{m}$ may also have jumps across the contours $\mathcal{C} \cap D_j$, $j = 1, \dots, 4$.
It is quite remarkable that all these jumps can be expressed completely in terms of $a$ and $b$ alone.

Let $\mathcal{B} = \bar{\mathcal{C}} \setminus \mathcal{C}$ denote the set of branch points of $\sqrt{4 - \Delta^2}$. The set $\mathcal{B}$ is always contained in the periodic spectrum $\mathcal{P}$ and it may be strictly smaller than $\mathcal{P}$ if $4 - \Delta^2$ has roots of even multiplicity. 
Let $\tilde{\Sigma} = \Sigma \cup \bar{\mathcal{C}}$ denote the union of the cross $\Sigma = \R \cup i\R$ and the set of branch cuts and branch points, see Figure \ref{contourfig}.
Let $\mathcal{S}$ denote the set of self-intersections of the contour $\tilde{\Sigma}$.
In the defocusing case, $\mathcal{S}$ only consists of the origin. In the focusing case, $\mathcal{S}$ consists of the origin together with any points at which vertical branch cuts intersect the real axis. 
Let $\tilde{\Sigma}_\star = \tilde{\Sigma} \setminus (\mathcal{B} \cup \mathcal{S})$ denote the contour $\tilde{\Sigma}$ with all branch points and all points of self-intersection removed.
The jump matrix $\tilde{v}$ is defined for $k \in \tilde{\Sigma}_\star$ as follows:
\begin{align}\label{tildevdef}
\tilde{v} = \begin{cases}  \tilde{v}_1 = \begin{pmatrix}  \frac{a  - \lambda b \bar{\tilde{\Gamma}} - \lambda \tilde{\Gamma} (\bar{a}\bar{\tilde{\Gamma}} - \bar{b})e^{2ikL} }{a- \lambda b \bar{\tilde{\Gamma}}} & -\tilde{\Gamma}  e^{2ikL} e^{-2 i \theta } \\
 \frac{\lambda \bar{\tilde{\Gamma}} e^{2 i \theta } }{(a - \lambda b \bar{\tilde{\Gamma}}) (a + \lambda\bar{b} \tilde{\Gamma}  e^{2ikL})} & \frac{a}{a + \lambda \bar{b} \tilde{\Gamma}  e^{2ikL}} \end{pmatrix}, & k \in i\R_+ \setminus \bar{\mathcal{C}},
	\\ 
\tilde{v}_2 =  \begin{pmatrix} 1 - \lambda \tilde{\Gamma} \bar{\tilde{\Gamma}} 
& -\frac{(\bar{a} \tilde{\Gamma} e^{2ikL}+b) e^{-2 i \theta }}{\bar{a}+ \lambda b \bar{\tilde{\Gamma}} e^{-2ikL}} \\
 \frac{\lambda (a \bar{\tilde{\Gamma}}e^{-2ikL}+\bar{b}) e^{2 i \theta } }{a+ \lambda \bar{b} \tilde{\Gamma}  e^{2ikL} } &
   \frac{1}{(a+ \lambda \bar{b} \tilde{\Gamma}  e^{2ikL}) (\bar{a} + \lambda b \bar{\tilde{\Gamma}} e^{-2ikL})}  \end{pmatrix}, & k \in \R_+ \setminus \bar{\mathcal{C}},
   	\\  
\tilde{v}_3 = \begin{pmatrix} \frac{\bar{a}  - \lambda \bar{b} \tilde{\Gamma} - \lambda \bar{\tilde{\Gamma}}(a\tilde{\Gamma} - b) e^{-2ikL}}{\bar{a} - \lambda \bar{b} \tilde{\Gamma} } & - \frac{\tilde{\Gamma} e^{-2 i \theta}}{(\bar{a} - \lambda \bar{b} \tilde{\Gamma}) (\bar{a} + \lambda b \bar{\tilde{\Gamma}} e^{-2ikL} )} \\
 \lambda \bar{\tilde{\Gamma}} e^{-2ikL} e^{2 i \theta} & \frac{\bar{a}}{\bar{a} + \lambda b \bar{\tilde{\Gamma}} e^{-2ikL}}  \end{pmatrix}, & k \in i\R_- \setminus \bar{\mathcal{C}},
   	\\ 
 \tilde{v}_4 = \begin{pmatrix}	
 \frac{1 - \lambda \tilde{\Gamma} \bar{\tilde{\Gamma}}}{(a - \lambda b \bar{\tilde{\Gamma}}) (\bar{a} - \lambda \bar{b} \tilde{\Gamma} )} & -\frac{(a \tilde{\Gamma} - b)e^{-2 i \theta } }{\bar{a} - \lambda \bar{b} \tilde{\Gamma} } \\
 \frac{\lambda(\bar{a} \bar{\tilde{\Gamma}} - \bar{b})e^{2 i \theta }}{a - \lambda b \bar{\tilde{\Gamma}}} & 1
 \end{pmatrix}, & k \in \R_- \setminus \bar{\mathcal{C}},
 	\\
\tilde{v}^{\cut}_{D_1} = 
\begin{pmatrix} \frac{a + \lambda \bar{b} \tilde{\Gamma}_+ e^{2ikL}}{a + \lambda \bar{b} \tilde{\Gamma}_- e^{2ikL}} & (\tilde{\Gamma}_- - \tilde{\Gamma}_+) e^{2ikL} e^{-2i\theta} \\ 0 & \frac{a + \lambda \bar{b} \tilde{\Gamma}_- e^{2ikL}}{a + \lambda \bar{b}\tilde{\Gamma}_+ e^{2ikL}} \end{pmatrix}, & k \in \mathcal{C} \cap D_1,
	\\
\tilde{v}^{\cut}_{D_2} = \begin{pmatrix} 1 & 0 \\ 
\frac{\lambda (\bar{\tilde{\Gamma}}_- - \bar{\tilde{\Gamma}}_+)e^{2i\theta}}{(a - \lambda b \bar{\tilde{\Gamma}}_-)(a - \lambda b \bar{\tilde{\Gamma}}_+)} & 1 \end{pmatrix},
& k \in \mathcal{C} \cap D_2,
	\\
\tilde{v}^{\cut}_{D_3} = \begin{pmatrix} 1 & \frac{(\tilde{\Gamma}_- - \tilde{\Gamma}_+)e^{-2i\theta}}{(\bar{a} - \lambda \bar{b} \tilde{\Gamma}_-)(\bar{a} - \lambda \bar{b} \tilde{\Gamma}_+)} \\
0 & 1 \end{pmatrix}, & k \in \mathcal{C} \cap D_3,
	\\
\tilde{v}^{\cut}_{D_4} = \begin{pmatrix} \frac{\bar{a} + \lambda b \bar{\tilde{\Gamma}}_- e^{-2ikL}}{\bar{a} + \lambda b \bar{\tilde{\Gamma}}_+ e^{-2ikL}} & 0 \\
\lambda(\bar{\tilde{\Gamma}}_- - \bar{\tilde{\Gamma}}_+) e^{2i\theta} e^{-2ikL} & \frac{\bar{a} + \lambda b \bar{\tilde{\Gamma}}_+ e^{-2ikL}}{\bar{a} + \lambda b \bar{\tilde{\Gamma}}_- e^{-2ikL}} \end{pmatrix}, & k \in \mathcal{C} \cap D_4,
	\\
\tilde{v}_1^{\cut} = \tilde{v}^{\cut}_{D_1} \tilde{v}_{1-},
& k \in \mathcal{C} \cap i\R_+,	
	\\
\tilde{v}_2^{\cut} = \tilde{v}^{\cut}_{D_1} \tilde{v}_{2-},
& k \in \mathcal{C} \cap \R_+,	
	\\
\tilde{v}_3^{\cut} = \tilde{v}^{\cut}_{D_3} \tilde{v}_{3-},
& k \in \mathcal{C} \cap i\R_-,	
	\\
\tilde{v}_4^{\cut} = \tilde{v}^{\cut}_{D_3} \tilde{v}_{4-},
& k \in \mathcal{C} \cap \R_-,
\end{cases}
\end{align}
where $\tilde{v}_{j-}$ denotes the boundary values of the matrix $\tilde{v}_j$ as $k$ approaches $\mathcal{C}$ from the right.
We emphasize that the jump matrix $\tilde{v}$ can be computed from the knowledge of the initial datum alone.

\begin{figure}
\begin{center}
\begin{overpic}[width=.45\textwidth]{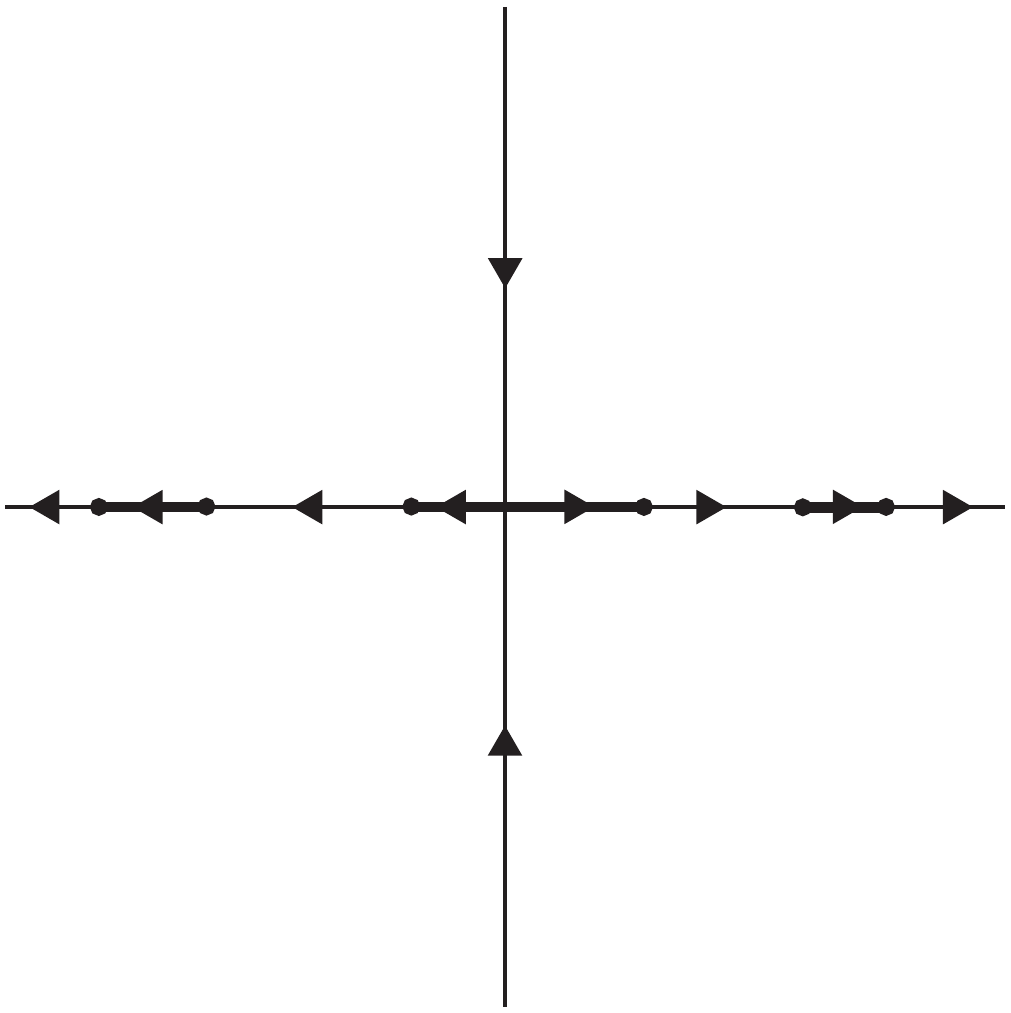}
      \put(52.5,72){\small $\tilde{v}_1$} 
      \put(92,53.5){\small $\tilde{v}_2$} 
      \put(54,43.5){\small $\tilde{v}_2^{\cut}$} 
      \put(68,53.5){\small $\tilde{v}_2$} 
      \put(80,43.5){\small $\tilde{v}_2^{\cut}$} 
      \put(52.5,25){\small $\tilde{v}_3$} 
      \put(3,53.5){\small $\tilde{v}_4$} 
      \put(12,43.5){\small $\tilde{v}_4^{\cut}$} 
      \put(29,53.5){\small $\tilde{v}_4$} 
      \put(41,43.5){\small $\tilde{v}_4^{\cut}$} 
     \end{overpic}\qquad
     \begin{overpic}[width=.45\textwidth]{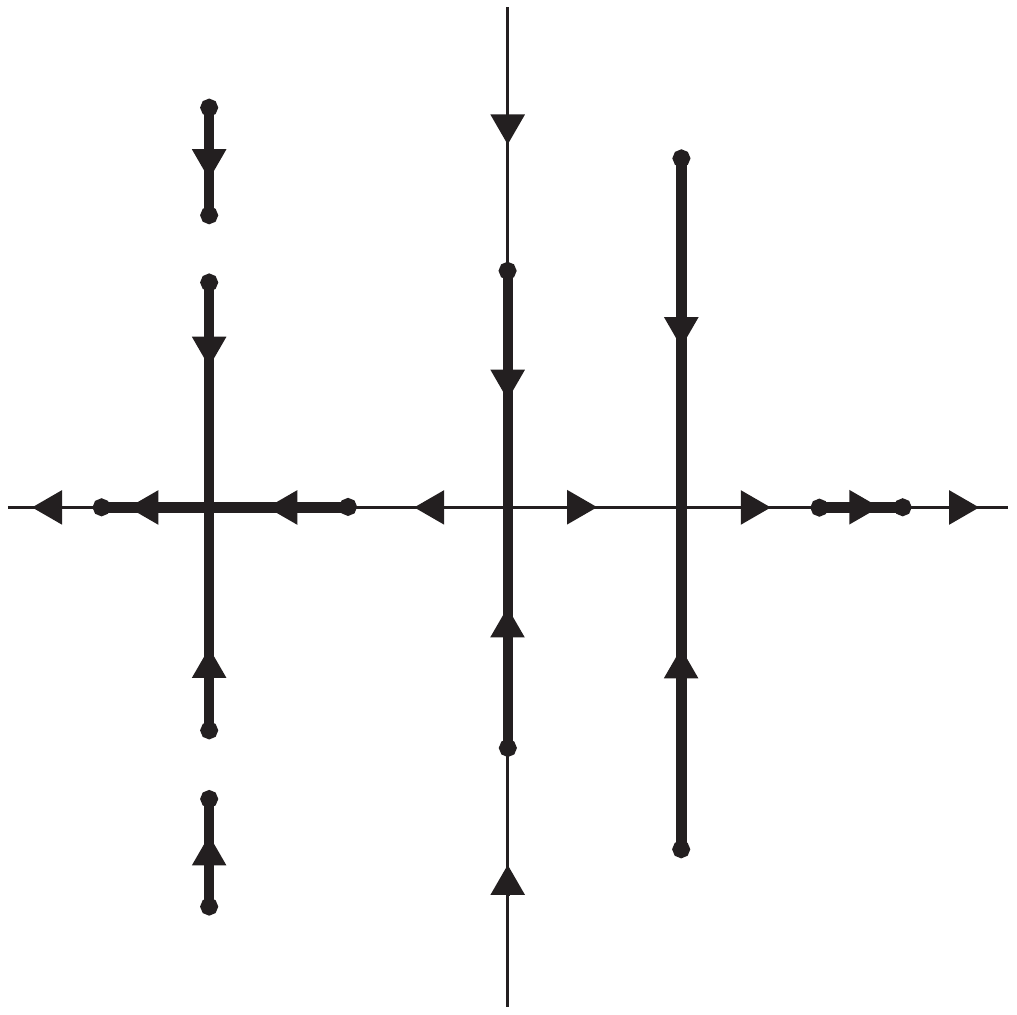}
      \put(52.5,86){\small $\tilde{v}_1$} 
      \put(52.5,62){\small $\tilde{v}_1^{\cut}$} 
      \put(69.5,67){\small $\tilde{v}_{D_1}^{\cut}$} 
      \put(55,53){\small $\tilde{v}_2$} 
      \put(72,53){\small $\tilde{v}_2$} 
      \put(82,53.5){\small $\tilde{v}_2^{\cut}$} 
      \put(23,83){\small $\tilde{v}_{D_2}^{\cut}$} 
      \put(23,65){\small $\tilde{v}_{D_2}^{\cut}$} 
      \put(23,33){\small $\tilde{v}_{D_3}^{\cut}$} 
      \put(23,14){\small $\tilde{v}_{D_3}^{\cut}$} 
      \put(93,53){\small $\tilde{v}_2$} 
      \put(52.5,12){\small $\tilde{v}_3$} 
      \put(52.5,37){\small $\tilde{v}_3^{\cut}$} 
      \put(2,53){\small $\tilde{v}_4$} 
      \put(11,53.5){\small $\tilde{v}_4^{\cut}$} 
      \put(26,53.5){\small $\tilde{v}_4^{\cut}$} 
      \put(69.5,33){\small $\tilde{v}_{D_4}^{\cut}$} 
      \put(41,53){\small $\tilde{v}_4$} 
     \end{overpic}
    \caption{\label{contourfig} Example of the contour $\tilde{\Sigma} = \R \cup i\R \cup \mathcal{C}$ in the complex $k$-plane and the associated jump matrices for RH problem \ref{RHmtilde} for $\lambda = 1$ (left) and $\lambda = -1$ (right). Dots indicate branch points and thick lines indicate the set of branch cuts $\mathcal{C}$.}
     \end{center}
\end{figure}
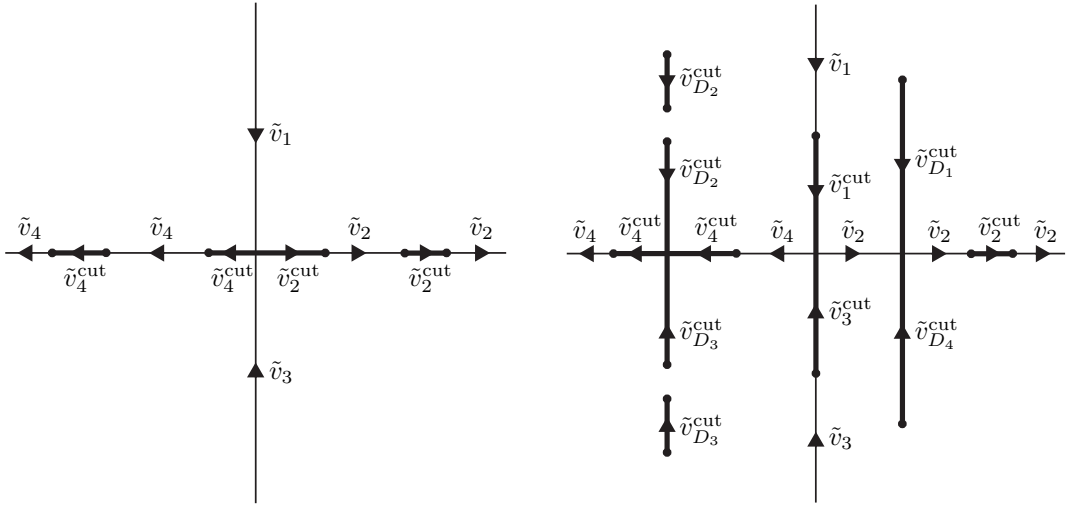

Whereas the function $m$, in general, has singularities at the poles of $\Gamma$, $\eta$, and $\xi$, it turns out that $\tilde{m}$ is analytic at these points. In fact, $\tilde{m}$ can have singularities only if $\tilde{\Gamma}$ has poles in the first or third quadrant. We make the following assumption. 

\begin{assumption}\label{tildepolesassumption}
We assume the following:
\begin{enumerate}[$-$]
\item $\tilde{\Gamma}$ has no poles on the contour $\tilde{\Sigma}$.

\item In $D_1$, $\tilde{\Gamma}(k)$ has at most finitely many poles $\{p_j\}_1^{n_1} \subset D_1$ and these poles are all simple.

\item In $D_3$, $\tilde{\Gamma}(k)$ has at most finitely many poles $\{q_j\}_{1}^{n_3} \subset D_2$ and these poles are all simple.
\end{enumerate}
\end{assumption}

 \begin{remark}\label{polesremark}
 Regarding assumption \ref{tildepolesassumption}, we note that there exist large families of initial conditions for which it can be shown explicitly that $\tilde{\Gamma}$ does not have poles. For example, the single exponential families considered in Section \ref{exponentialexamplesec} and in \cite{FLDxperiodic} are of this type. More precisely, the definition (\ref{tildeGammadef}) of $\tilde{\Gamma}$ has the form $\tilde{\Gamma} = \tilde{\gamma}/\bar{b}$, where the numerator $\tilde{\gamma}$ has no poles. Thus any possible poles of $\tilde{\Gamma}$ are generated by the zeros of $\bar{b}$. However, in many cases the zeros of $\bar{b}$ are cancelled by zeros of $\tilde{\gamma}$. Indeed, suppose $k_0$ is a zero of $\bar{b}$. Using the relation $a \bar{a} - \lambda b \bar{b} = 1$, we can write
 \begin{align}\label{tildegammadef}
\tilde{\gamma}(k) = \frac{\lambda(\bar{a} e^{ikL} - a e^{-ikL})}{2 e^{ikL}} 
\bigg(1 - \sqrt{1 + \frac{4\lambda b \bar{b}}{(\bar{a} e^{ikL} - a e^{-ikL})^2}}\bigg),
\end{align}
where the branch of the square root is fixed so that it tends to $1$ as $k$ tends to infinity in $\C  \setminus \cup_{n\in \Z} \mathcal{D}_n$ (it can be shown that $4\lambda b \bar{b}/(\bar{a} e^{ikL} - a e^{-ikL})^2 = O(k^{-2})$ in this limit as a consequence of (\ref{abasymptotics})).
Thus, assuming that $\bar{a} e^{ikL} - a e^{-ikL}$ is nonzero at $k_0$, we see that $\tilde{\gamma}$ vanishes at $k_0$ to the same order as $\bar{b}$ whenever the branch of the square root in (\ref{tildegammadef}) is such that it is close to $1$ for $k$ near $k_0$. It follows, in particular, that $\tilde{\Gamma}(k)$ cannot have poles whenever $k \in \C \setminus \cup_{n \in \Z} \mathcal{D}_n$ is large enough.
The question of whether $\tilde{\Gamma}$ is always pole-free is under consideration. For the family of single exponential initial data considered in Section \ref{exponentialexamplesec}, we will see that $\tilde{\Gamma}$ has no poles even though $\bar{b}$ has infinitely many zeros. 
\end{remark}

Let $\tilde{P} = \{p_j, \bar{p}_j\}_1^{n_1} \cup \{q_j, \bar{q}_j\}_1^{n_2}$ denote the set of poles of $\tilde{\Gamma}$ in $D_1 \cup D_3$ and their complex conjugates.
We will show that $\tilde{m}$ is the unique solution of the following RH problem.

\begin{RHproblem}{\bf(The RH problem for $\tilde{m}$)}\label{RHmtilde}
Find a $2 \times 2$-matrix valued function $\tilde{m}(x,t,k)$ with the following properties:
\begin{itemize}
\item $\tilde{m}(x,t,\cdot) : \C \setminus (\tilde{\Sigma} \cup \tilde{P}) \to \C^{2 \times 2}$ is analytic.

\item The limits of $m(x,t,k)$ as $k$ approaches $\tilde{\Sigma}_\star$ from the left and right exist, are continuous on $\tilde{\Sigma}_\star$, and satisfy
\begin{align}\label{mtildejump}
  \tilde{m}_-(x,t,k) = \tilde{m}_+(x, t, k) \tilde{v}(x, t, k), \qquad k \in \tilde{\Sigma}_\star.
\end{align}

\item $\tilde{m}(x,t,k) = I  + O\big(k^{-1}\big)$ as $k \to \infty$, $k \in \C \setminus \cup_{n\in \Z} \mathcal{D}_n$, where $\mathcal{D}_n$ is the open disk of radius $\pi/(4L)$ centered at $n\pi/L$ defined in (\ref{Dndef}).

\item $\tilde{m}(x,t,k) = O(1)$ as $k \to \tilde{\Sigma} \setminus \tilde{\Sigma}_\star$, $k \in \C \setminus \tilde{\Sigma}$.

\item At the points $p_j \in D_1$  and  $\bar{p}_j \in D_4$, $\tilde{m}$ has at most simple poles and the residues at these poles satisfy, for $j = 1, \dots, n_1$,
\begin{subequations}\label{pqresidues}
\begin{align}
& \underset{k = p_j}{\res} [\tilde{m}(x,t,k)]_2 =  \Big\{[\tilde{m}]_1 (a + \lambda \bar{b} \tilde{\Gamma} e^{2ikL})\bar{a} e^{2ikL} e^{-2i\theta} \Big\}(x,t,p_j) \, \underset{k=p_j}{\res} \tilde{\Gamma}(k),
	\\
& \underset{k=\bar{p}_j}{\res} [\tilde{m}(x,t,k)]_1 = \Big\{\lambda [\tilde{m}]_2 (\bar{a} + \lambda b \bar{\tilde{\Gamma}} e^{-2ikL}) a e^{-2ikL} e^{2i\theta} \Big\}(x,t,\bar{p}_j) \, \overline{\underset{k=p_j}{\res} \tilde{\Gamma}(k)}.
\end{align}

\item At the points $q_j \in D_3$  and  $\bar{q}_j \in D_2$, $\tilde{m}$ has at most simple poles and the residues at these poles satisfy, for $j = 1, \dots, n_2$,
\begin{align}
& \underset{k=q_j}{\res} [\tilde{m}(x,t,k)]_2 = \bigg\{ [\tilde{m}]_1 \frac{a e^{-2i\theta}}{\bar{a} - \lambda \bar{b} \tilde{\Gamma}}\bigg\}(x,t,q_j) \, \underset{k=q_j}{\res} \tilde{\Gamma}(k),
	\\
& \underset{k=\bar{q}_j}{\res} [\tilde{m}(x,t,k)]_1 = \bigg\{\lambda [\tilde{m}]_2 \frac{\bar{a} e^{2i\theta}}{a - \lambda b \bar{\tilde{\Gamma}}}\bigg\}(x,t,\bar{q}_j) \, \overline{\underset{k=q_j}{\res} \tilde{\Gamma}(k)}.
\end{align}
\end{subequations}

\end{itemize}
\end{RHproblem}

In order to show that $\tilde{m}$ and $\tilde{v}$ have the appropriate regularity properties, we need the following lemma which shows that some of the denominators in the definition (\ref{gdef}) of the matrices $g_j(k)$ are nowhere zero.

\begin{lemma}\label{Deltaidentitieslemma}
The following identities are valid:
\begin{align}\label{Deltaidentity}
a + \lambda \bar{b} \tilde{\Gamma} e^{2ikL} = a - \lambda b \bar{\tilde{\Gamma}}  = \frac{e^{ikL}}{2}\big(\Delta - i\sqrt{4 - \Delta^2}\big).
\end{align}
In particular, the function $a + \lambda \bar{b} \tilde{\Gamma} e^{2ikL} = a - \lambda b \bar{\tilde{\Gamma}}$ is nonzero for all $k \in \C \setminus \mathcal{C}$.
\end{lemma}
\begin{proof}
The identities follow by a direct computation using the definitions (\ref{tildeGammadef}) and (\ref{Deltadef}) of $\tilde{\Gamma}$ and $\Delta$. If $\Delta - i\sqrt{4 - \Delta^2} = 0$ at some $k$, then $\Delta^2 = \Delta^2 - 4$ at $k$ which is a contradiction. Thus $a + \lambda \bar{b} \tilde{\Gamma} e^{2ikL} = a - \lambda b \bar{\tilde{\Gamma}}$ has no zeros.
\end{proof}

\subsection{Main result}
The following theorem, which is the main result of the paper, provides an expression for the solution $q(x,t)$ of the $x$-periodic NLS equation in terms of the solution of the RH problem \ref{RHmtilde}. Since the formulation of this RH problem only involves quantities defined in terms of the initial datum, the theorem provides an effective solution of the IVP for the $x$-periodic NLS equation.

\begin{theorem}\label{mainth}
Suppose $q(x,t)$ is a smooth solution of (\ref{NLS}) for $(x,t) \in \R \times [0,\infty)$ which is $x$-periodic of period $L > 0$, i.e., $q(x+L, t) = q(x,t)$.
Define $a(k)$ and $b(k)$ by (\ref{abdef}) and let $\tilde{\Gamma}(k)$ be the function defined in terms of $a$ and $b$ by (\ref{tildeGammadef}). 
Suppose Assumption \ref{tildepolesassumption} holds.
 
Then the RH problem \ref{RHmtilde} has a unique solution $\tilde{m}(x,t,k)$ for each $(x,t) \in [0,L] \times [0,\infty)$ and the solution $q$ can be obtained from  $\tilde{m}$ via the relation
\begin{align}\label{recoverq}
  q(x,t) = 2i\lim_{k \to \infty} k \tilde{m}_{12}(x,t,k), \qquad (x,t) \in [0,L] \times [0,\infty),
\end{align}
where the limit is taken along any ray $\{k | \arg k = \phi\}$ where $\phi \in \R \setminus \{n\pi/2 \, | \, n \in \Z\}$ (i.e., the ray is not contained in $\R \cup i\R$).
\end{theorem}
\begin{proof}
The function $\tilde{\Gamma}$ is analytic in $(D_1 \cup D_3) \setminus (\tilde{\Sigma} \cup \tilde{P})$ with continuous boundary values on $\tilde{\Sigma}_\star$. 
Thus it follows from Lemma \ref{Deltaidentitieslemma} that each of the matrices appearing on the right-hand side of (\ref{tildevdef}) is well-defined and continuous on its domain of definition. In particular, $\tilde{v}$ is well-defined and piecewise continuous on $\tilde{\Sigma}$. 

Let us prove uniqueness of $\tilde{m}$. 
The problem with a nonempty set $\tilde{P}$ can be transformed into a problem for which $\tilde{P}$ is empty following a standard procedure, see e.g. \cite{FI1996}; we may therefore assume that $\tilde{P}$ is empty when proving uniqueness.
Suppose $\tilde{m}$ is a solution of the RH problem \ref{RHmtilde}. We will first show that $\tilde{m}$ has unit determinant. The jump matrix $\tilde{v}$ has unit determinant everywhere on $\tilde{\Sigma}_\star$. Thus $\det \tilde{m}$ is an entire function except for possible singularities at points in the discrete set $\tilde{\Sigma} \setminus \tilde{\Sigma}_\star$. However, the assumption that $\tilde{m} = O(1)$ as $k \to \tilde{\Sigma} \setminus \tilde{\Sigma}_\star$ implies that these singularities are removable. Thus $\det \tilde{m}$ is an entire function. 
The assumption that $\tilde{m} = I  + O(k^{-1})$ as $k \in \C \setminus \cup_{n\in \Z} \mathcal{D}_n$ tends to infinity implies that there is a constant $C > 0$ such that $|\tilde{m} - I| \leq C/k$ on each of the circles $|k| = (n + 1/2)\pi /L$, $n = 1, 2, \dots$. Hence $\det \tilde{m} = 1 + O(k^{-1})$ uniformly on these circles. By the maximum modulus principle, we conclude that $\det \tilde{m} = 1 + O(k^{-1})$ as $k \to \infty$. Hence, by Liouville's theorem, $\det \tilde{m} = 1$ for all $k \in \C$. 
 
Assume $\tilde{m}$ and $\tilde{n}$ are two solutions of the RH problem \ref{RHmtilde}. 
Since $\tilde{n}$ has unit determinant, the function $\tilde{m} \tilde{n}^{-1}$ is well-defined on $\C \setminus \tilde{\Sigma}$ and has continuous boundary values on $\tilde{\Sigma}_\star$ which satisfy
$$(\tilde{m} \tilde{n}^{-1})_- = \tilde{m}_+ \tilde{v} \tilde{v}^{-1} \tilde{n}_+^{-1} = (\tilde{m} \tilde{n}^{-1})_+.$$
Thus $\tilde{m} \tilde{n}^{-1}$ is an entire function except for possible singularities at points in the discrete set $\tilde{\Sigma} \setminus \tilde{\Sigma}_\star$.
The same arguments that led to $\det \tilde{m} = 1$, show that these singularities are removable and that $\tilde{m} \tilde{n}^{-1}$ is in fact identically equal to the identity matrix. This proves uniqueness. 

Fix $(x,t) \in [0,L] \times [0,\infty)$. Choose $T \in (t, \infty)$ and define $\tilde{m}$ by (\ref{tildemdef}) with $m$ and $g$ defined using $T$ as final time.
We will show that  $\tilde{m}$ satisfies the RH problem \ref{RHmtilde} and that (\ref{recoverq}) holds at the point $(x,t)$.

The function $m$ obeys the symmetries (\ref{msymm}) and it is easy to check that $g$ satisfies the same symmetries:
\begin{align}\label{gsymm}
g(k) = \begin{cases} \sigma_1 \overline{g(\bar{k})} \sigma_1, & \lambda = 1, \\
\sigma_3 \sigma_1 \overline{g(\bar{k})}\sigma_1 \sigma_3, & \lambda = -1.
\end{cases}
\end{align}
It follows that $\tilde{m}$ also obeys these symmetries:
\begin{align}\label{mtildesymm}
\tilde{m}(x,t,k) = \begin{cases} \sigma_1 \overline{\tilde{m}(x,t,\bar{k})} \sigma_1, & \lambda = 1, \\
\sigma_3 \sigma_1 \overline{\tilde{m}(x,t,\bar{k})}\sigma_1 \sigma_3, & \lambda = -1.
\end{cases}
\end{align}
Moreover, since $m$  and $g$ have unit determinant, we have $\det \tilde{m} = 1$.

Let us show that $\tilde{m}$  is analytic for $k \in \C \setminus (\tilde{\Sigma} \cup \tilde{P})$.
This can be established by considering the analyticity properties of $m$ and using the conditions (\ref{Gammaresidues})-(\ref{xiresidues}) to show that $\tilde{m} = mg$ has no poles at these points. However, then we would have to restrict ourselves to initial data for which Assumption \ref{poleassumption} holds. We therefore instead give a direct argument which takes the eigenfunctions $\mu_j$ as its starting point. 

In light of the symmetries (\ref{mtildesymm}), it is enough to establish analyticity of $\tilde{m}$ for $k \in \C^+ \setminus (\tilde{\Sigma} \cup \tilde{P})$, where $\C^+ = \{\im k > 0\}$ denotes the open upper half-plane.
The definitions (\ref{mdef}) and (\ref{tildemdef}) imply that $\tilde{m}_1$ and $\tilde{m}_2$ can be expressed in terms of the eigenfunctions $\mu_j$ as follows:
\begin{subequations}\label{tildemmu}
\begin{align}\label{tildemmua}
& \tilde{m}_1 = \bigg(\frac{[\mu_2]_1}{a + \lambda \bar{b} \tilde{\Gamma} e^{2ikL}}, 
\frac{[\mu_2]_1 A(\tilde{\Gamma} - \Gamma)e^{2ikL}e^{-2i\theta}
+ [\mu_4]_2 (a + \lambda \bar{b} \tilde{\Gamma} e^{2ikL})}{\alpha}\bigg),
	\\\label{tildemmub}
& \tilde{m}_2 = \bigg(\frac{[\mu_1]_1}{d}
+  \frac{[\mu_3]_2 \lambda(\bar{\tilde{\Gamma}} -\bar{\Gamma}) e^{2i\theta}}{(a - \lambda b \bar{\Gamma})(a - \lambda b \bar{\tilde{\Gamma}})},  [\mu_3]_2\bigg).
\end{align}
\end{subequations}
Utilizing (\ref{mu42mu21}) to simplify the second column of $\tilde{m}_1$ and (\ref{mu11mu31}) to simplify the first column of $\tilde{m}_2$, we can write (\ref{tildemmu}) as
\begin{subequations}\label{tildemmu2}
\begin{align}\label{tildemmu2a}
& \tilde{m}_1 = \bigg(\frac{[\mu_2]_1}{a + \lambda \bar{b} \tilde{\Gamma} e^{2ikL}}, 
[\mu_2]_1(\bar{a} \tilde{\Gamma} e^{2ikL} + b) e^{-2i\theta}
+  [\mu_2]_2 (a + \lambda \bar{b} \tilde{\Gamma} e^{2ikL})\bigg),
	\\\label{tildemmu2b}
& \tilde{m}_2 = \bigg([\mu_3]_1 + [\mu_3]_2 \frac{\lambda(\bar{a} \bar{\tilde{\Gamma}} - \bar{b})}{a - \lambda b \bar{\tilde{\Gamma}}} e^{2i\theta},  [\mu_3]_2\bigg).
\end{align}
\end{subequations}
The functions $\{\mu_j\}_1^4$ are entire and $\tilde{\Gamma}$ is analytic in $(D_1 \cup D_3) \setminus (\mathcal{C} \cup \tilde{P})$. 
Furthermore, it was shown in Lemma \ref{Deltaidentitieslemma} that $a + \lambda \bar{b} \tilde{\Gamma} e^{2ikL} = a - \lambda b \bar{\tilde{\Gamma}}$ is nowhere zero.
Thus it follows from (\ref{mtildesymm}) and (\ref{tildemmu2}) that $\tilde{m}$ is an analytic function of $k \in \C \setminus (\tilde{\Sigma} \cup \tilde{P})$ with continuous boundary values on $\tilde{\Sigma}_\star$.
Recalling Assumption \ref{tildepolesassumption}, it also follows from (\ref{mtildesymm}) and (\ref{tildemmu2}) that $\tilde{m}(x,t,k) = O(1)$ as $k \to \tilde{\Sigma} \setminus \tilde{\Sigma}_\star$, $k \in \C \setminus \tilde{\Sigma}$, and that $\tilde{m}$ obeys the residue conditions (\ref{pqresidues}) at the points in $\tilde{P}$.

We next show that the boundary values of $\tilde{m}$ satisfy the jump relation (\ref{mtildejump}). 
We already saw from (\ref{tildevgv}) that (\ref{mtildejump}) holds on $(\R \cup i\R) \setminus \mathcal{C}$. To compute the jump across $\mathcal{C} \cap D_j$ for $j = 1, \dots, 4$,  we note that $\tilde{m} = mg$ and that $m$ has no jump across $\mathcal{C} \cap D_j$. Hence $\tilde{m}_- = \tilde{m}_+ \tilde{v}_{D_j}^{\cut}$ on $\mathcal{C} \cap D_j$ where
$$\tilde{v}_{D_j}^{\cut} = g_{j+}^{-1}g_{j-}, \qquad k \in \mathcal{C} \cap D_j.$$
A direct computation using the definition (\ref{gdef}) of $g_j$ shows that the matrix $\tilde{v}^{\cut}_{D_j}$ is given by the expression in (\ref{tildevdef}).
On the other hand, on the part of $\mathcal{C}$ that is contained in $\Sigma = \R \cup i\R$, we have 
$$\tilde{v} = \tilde{m}_+^{-1} \tilde{m}_-
= (m_+g_+)^{-1}(m_-g_-)
= g_+^{-1} v g_-, \qquad k \in \mathcal{C} \cap \Sigma.$$
It follows that
$$\tilde{v}
= \begin{cases}
\tilde{v}^{\cut}_1 = g_{1+}^{-1} v_1 g_{2-}
= g_{1+}^{-1} g_{1-} \tilde{v}_{1-}
= \tilde{v}^{\cut}_{D_1} \tilde{v}_{1-},
& k \in \mathcal{C} \cap i\R_+,
	\\
\tilde{v}^{\cut}_2 = g_{1+}^{-1} v_2 g_{4-}
= g_{1+}^{-1} g_{1-} \tilde{v}_{2-}
= \tilde{v}^{\cut}_{D_1} \tilde{v}_{2-},
& k \in \mathcal{C} \cap \R_+,
	\\
\tilde{v}^{\cut}_3 = g_{3+}^{-1} v_3 g_{4-}
= g_{3+}^{-1} g_{3-} \tilde{v}_{3-}
= \tilde{v}^{\cut}_{D_3} \tilde{v}_{3-},
& k \in \mathcal{C} \cap i\R_-,
	\\
\tilde{v}^{\cut}_4 = g_{3+}^{-1} v_4 g_{2-}
= g_{3+}^{-1} g_{3-} \tilde{v}_{4-}
= \tilde{v}^{\cut}_{D_3} \tilde{v}_{4-},
& k \in \mathcal{C} \cap \R_-.
\end{cases}
$$
This completes the proof of the jump relation (\ref{mtildejump}).

It only remains to show that $\tilde{m} = I  + O(k^{-1})$ as $k \in \C \setminus \cup_{n\in \Z} \mathcal{D}_n$ approaches infinity. 
This follows from the fact that $m = I + O(k^{-1})$ as $k \to \infty$ provided that we can show that 
\begin{align}\label{gasymptotics}
  g(k) = I + O\big(k^{-1}\big)\quad \text{as $k \to \infty$, $k \in \C \setminus \cup_{n\in \Z} \mathcal{D}_n$}.
\end{align}
In fact, due to the symmetry (\ref{gsymm}) of $g$, it is enough to prove for $j = 1,2$ that
\begin{align}\label{gjasymptotics}
  g_j(k) = I + O\big(k^{-1}\big)\quad \text{as $k \to \infty$, $k \in \bar{D}_j \setminus \cup_{n\in \Z} \mathcal{D}_n$}.
\end{align}
The estimates (\ref{abasymptotics}) imply
\begin{align}\nonumber
4 - \Delta^2
= &\; 4 - a^2 e^{-2ikL} - 2 a \bar{a} - \bar{a}^2 e^{2ikL}
= 4 - \bigg(1 + O\bigg(\frac{1}{k}\bigg) + O\bigg(\frac{e^{2ikL}}{k}\bigg)\bigg)^2 e^{-2ikL} 
	\\\nonumber
& - 2 \bigg(1 + O\bigg(\frac{1}{k}\bigg) + O\bigg(\frac{e^{2ikL}}{k}\bigg)\bigg)\bigg(1 + O\bigg(\frac{1}{k}\bigg) + O\bigg(\frac{e^{-2ikL}}{k}\bigg)\bigg)
	\\\nonumber
& - \bigg(1 + O\bigg(\frac{1}{k}\bigg) + O\bigg(\frac{e^{-2ikL}}{k}\bigg)\bigg)^2 e^{2ikL}
	\\\nonumber
= &\; 4 - e^{-2ikL} - 2 - e^{2ikL} + O\bigg(\frac{1}{k}\bigg) + O\bigg(\frac{e^{2ikL}}{k}\bigg) + O\bigg(\frac{e^{-2ikL}}{k}\bigg) 
	\\ \label{4minusDelta2}
= &\; 4\sin^2(kL) + O\bigg(\frac{1}{k}\bigg) + O\bigg(\frac{e^{2ikL}}{k}\bigg) + O\bigg(\frac{e^{-2ikL}}{k}\bigg),
\qquad k \to \infty, \ k \in \C. 
\end{align}
Using that (see e.g. \cite[Lemma F.2]{GK2014})
\begin{align}\label{sinbound}
|\sin{kL}| > e^{|\im k| L}/4 \quad \text{for $k \in \C \setminus \cup_{n\in \Z} \mathcal{D}_n$},
\end{align}
it follows that
\begin{align}\nonumber
4 - \Delta^2
& = 4\sin^2(kL)\bigg\{1 + O\bigg(\frac{1 + e^{-2L\im k} + e^{2L \im k}}{k e^{2L|\im k|}}\bigg) \bigg\}
	\\\label{4Deltaasymptotics}
& = 4\sin^2(kL)\big(1 + O(k^{-1})\big),
\qquad k \to \infty, \  k \in \C  \setminus \cup_{n\in \Z} \mathcal{D}_n,
\end{align}
and hence
\begin{align}\label{sqrt4Deltaasymptotics}
\sqrt{4 - \Delta^2}
& = 2\sin(kL)\big(1 + O(k^{-1})\big), \qquad k \to \infty, \  k \in \C  \setminus \cup_{n\in \Z} \mathcal{D}_n.
\end{align}
In particular, as $k \to \infty$ in the closed upper half-plane $\bar{\C}^+ = \{\im k \geq 0\}$, we have
\begin{align}\label{esqrtasymptotics}
-i\frac{e^{ikL}}{2}\sqrt{4 - \Delta^2} = -ie^{ikL}\sin(kL) + O(k^{-1}), \qquad k \to \infty, \ k \in \bar{\C}^+ \setminus \cup_{n\in \Z} \mathcal{D}_n,
\end{align}

The estimates (\ref{abasymptotics}) also imply that $a = 1 + O(k^{-1})$ and $\bar{a} e^{2ikL} = e^{2ikL} + O(k^{-1})$ as $k \to \infty$ in $\bar{\C}^+$; thus
\begin{align}\label{eDeltaasymptotics}
\frac{e^{ikL}}{2}\Delta
 = \frac{a +\bar{a} e^{2ikL}}{2} 
 = e^{ikL} \cos(kL) + O(k^{-1}), \qquad k \to \infty, \ k \in \bar{\C}^+.
\end{align}
Combining (\ref{esqrtasymptotics}) and (\ref{eDeltaasymptotics}), we obtain
\begin{align*}
\frac{e^{ikL}}{2}\big(\Delta - i\sqrt{4 - \Delta^2}\big)
 = 1 + O(k^{-1}), \qquad k \to \infty, \ k \in \bar{\C}^+ \setminus \cup_{n\in \Z} \mathcal{D}_n.
\end{align*}
Recalling the identities in Lemma \ref{Deltaidentitieslemma}, this shows that
\begin{align}\label{abtildeGammaasymptotics}
a + \lambda \bar{b} \tilde{\Gamma} e^{2ikL} = a - \lambda b \bar{\tilde{\Gamma}}  = 1 + O(k^{-1}), \qquad k \to \infty, \ k \in \bar{\C}^+ \setminus \cup_{n\in \Z} \mathcal{D}_n.
\end{align}

Let us consider $g_1$. As $k \to \infty$ in $\bar{D}_1$, we have $\Gamma = O(k^{-1})$ and $\bar{b}e^{2ikL} = O(k^{-1})$, and hence $a + \lambda \bar{b} \Gamma e^{2ikL} = 1 + O(k^{-1})$. Together with (\ref{abtildeGammaasymptotics}), this yields
$$(g_1)_{11} = (g_1)_{22}^{-1} = \frac{a + \lambda \bar{b} \Gamma e^{2ikL}}{a + \lambda \bar{b} \tilde{\Gamma} e^{2ikL}}
 = 1 + O(k^{-1}), \quad k \to \infty, \ k \in \bar{D}_1 \setminus \cup_{n\in \Z} \mathcal{D}_n,$$
showing (\ref{gjasymptotics}) for the diagonal elements of $g_1$. 
As for the nonzero off-diagonal element 
$$(g_1)_{12} = (\tilde{\Gamma} - \Gamma) e^{2ik(L-x)} e^{-4ik^2t},$$
we note that solving the global relation (\ref{periodicgr}) for $\Gamma = B/A$ gives
\begin{align}\label{Gammasolvedfor}
\Gamma = \frac{\lambda}{2 e^{ikL} \bar{b}} \bigg[\bar{a} e^{ikL} - a e^{-ikL} - i\sqrt{4 - \Delta^2 - \frac{4\lambda \bar{b} e^{4ik^2T} c^+}{A^2}} \bigg],
\end{align}
where the branch of the square root is fixed by the requirement (cf. (\ref{sqrt4Delta2}))
$$\sqrt{4 - \Delta^2 - \frac{4\lambda \bar{b} e^{4ik^2T} c^+}{A^2}} \sim 2\sin(kL), \qquad k \to \infty, \ k \in \bar{D}_1  \setminus \cup_{n\in \Z} \mathcal{D}_n.$$
By (\ref{cplusasymptotics}), (\ref{abasymptotics}), (\ref{ABasymptotics}), (\ref{sinbound}), and (\ref{4Deltaasymptotics}), we have
\begin{align}\label{becplusA2}
\frac{4\lambda \bar{b} e^{4ik^2T} c^+}{A^2(4 - \Delta^2)} 
 =
 O\bigg(\frac{e^{4ik^2T}}{k^2}\bigg), \qquad k \to \infty, \ k \in (\bar{D}_1 \cup \bar{D}_3) \setminus \cup_{n\in \Z} \mathcal{D}_n;
\end{align}
thus the branch cuts and the values of the square root in (\ref{Gammasolvedfor}) are close to those of $\sqrt{4 - \Delta^2}$ for large $k \in (\bar{D}_1 \cup \bar{D}_3) \setminus \cup_{n\in \Z} \mathcal{D}_n$. In particular, both of these roots are analytic for large $k \in (\bar{D}_1 \cup \bar{D}_3) \setminus \cup_{n\in \Z} \mathcal{D}_n$. Subtracting (\ref{Gammasolvedfor}) from (\ref{tildeGammadef}), we find
\begin{align*}
\tilde{\Gamma} - \Gamma
= \frac{\lambda}{2 i e^{ikL} \bar{b}} \bigg[ \sqrt{4 - \Delta^2} - \sqrt{4 - \Delta^2 - \frac{4\lambda \bar{b} e^{4ik^2T} c^+}{A^2}} \bigg], \qquad k \in (\bar{D}_1 \cup \bar{D}_3)  \setminus \cup_{n\in \Z} \mathcal{D}_n.
\end{align*}
Utilizing (\ref{cplusasymptotics}), (\ref{sqrt4Deltaasymptotics}), and (\ref{sinbound}), we infer that, as $k \in \bar{D}_1  \setminus \cup_{n\in \Z} \mathcal{D}_n$ approaches infinity,
\begin{align}\nonumber
\tilde{\Gamma} - \Gamma
& = \frac{\lambda \sqrt{4 - \Delta^2}}{2 i e^{ikL} \bar{b}} \bigg[1  - \sqrt{1 - \frac{4\lambda \bar{b} e^{4ik^2T} c^+}{A^2 (4 - \Delta^2)}} \bigg]
	\\\nonumber
& = \frac{\lambda \sin(kL)(1 + O(k^{-1}))}{i e^{ikL} \bar{b}} O\bigg(\frac{\bar{b} e^{4ik^2T}}{k \sin^2(kL)(1 + O(k^{-1}))} \bigg)
	\\\label{tildeGammaminusGamma}
& = O\bigg(\frac{\sin(kL)}{e^{ikL}} \frac{e^{4ik^2T}}{k \sin^2(kL)} \bigg)
 = O\bigg(\frac{e^{4ik^2T}}{k}  \bigg).
\end{align}
Since $x \leq L$ and $T > t$, this yields
\begin{align}\label{g112asymptotics}
(g_1)_{12} = O\bigg(\frac{e^{2ik(L-x)} e^{4ik^2(T-t)}}{k}\bigg)
= O\bigg(\frac{1}{k}\bigg), \qquad k \to \infty, \ k \in \bar{D}_1  \setminus \cup_{n\in \Z} \mathcal{D}_n.
\end{align}
This completes the proof of (\ref{gjasymptotics}) for $j = 1$. 

We next consider $g_2$. We have
$$(g_2)_{21} = \frac{\lambda (\bar{\tilde{\Gamma}} - \bar{\Gamma})e^{2ikx + 4ik^2 t}}{(a - \lambda b \bar{\Gamma})(a - \lambda b \bar{\tilde{\Gamma}})}.$$
Since $\Gamma = O(k^{-1})$ as $k \in \bar{D}_3$ approaches infinity, (\ref{abasymptotics}) gives
$$a - \lambda b \bar{\Gamma} = 1 + O(k^{-1}), \qquad k \to \infty, \ k \in \bar{D}_2.$$
Moreover, by (\ref{abtildeGammaasymptotics}), $a - \lambda b \bar{\tilde{\Gamma}}  = 1 + O(k^{-1})$ as $k \in \bar{D}_2 \setminus \cup_{n\in \Z} \mathcal{D}_n$ goes to infinity.
On the other hand, proceeding as in (\ref{tildeGammaminusGamma}), we deduce that, as $k \in \bar{D}_3  \setminus \cup_{n\in \Z} \mathcal{D}_n$ approaches infinity,
\begin{align*}
\tilde{\Gamma} - \Gamma
& = \frac{\lambda \sqrt{4 - \Delta^2}}{2 i e^{ikL} \bar{b}} \bigg[1  - \sqrt{1 - \frac{4\lambda \bar{b} e^{4ik^2T} c^+}{A^2 (4 - \Delta^2)}} \bigg]
 = O\bigg(\frac{\sin(kL) }{e^{ikL} \bar{b}} \frac{\bar{b} e^{4ik^2T} c^+}{A^2\sin^2(kL)} \bigg)
	\\
& = O\bigg(\frac{e^{4ik^2T} c^+}{e^{ikL} \sin(kL)} \bigg)
 = O\bigg(\frac{e^{4ik^2T}}{k} \bigg).
\end{align*}
Since $x \geq 0$ and $T > t$, we conclude that
\begin{align}\label{g221asymptotics}
(g_2)_{21} = O\bigg(\frac{e^{2ikx + 4ik^2(t-T)}}{k} \bigg)
= O\bigg(\frac{1}{k}\bigg), \qquad k \to \infty, \ k \in \bar{D}_2  \setminus \cup_{n\in \Z} \mathcal{D}_n.
\end{align}
This completes the proof of (\ref{gasymptotics}).

Finally, by (\ref{gsymm}), (\ref{g112asymptotics}), and (\ref{g221asymptotics}), the difference $g - I$ is exponentially small as $k \to \infty$ along any ray $\{k | \arg k = \phi\}$ with $\phi \in \R \setminus \{n\pi/2 \, | \, n \in \Z\}$. Since $\tilde{m} = mg$, equation (\ref{recoverq}) then follows from (\ref{recoverqm}).
 This completes the proof of the theorem. 
\end{proof}

\section{Example: A single exponential}\label{exponentialexamplesec}
We illustrate the approach of Theorem \ref{mainth} by considering the following initial datum involving a single exponential:
\begin{align}\label{singleexponentialq0}
q(x,0) = q_0 e^{\frac{2i\pi N}{L} x}, \qquad x \in [0,L],
\end{align}
where $N$ is an integer and the constant $q_0 > 0$ can be taken to be positive due to the phase invariance of (\ref{NLS}). 

\subsection{RH problem for $\tilde{m}$}
Let $N$ be an integer. Direct integration of the $x$-part of the Lax pair (\ref{laxpair}) with $q$ given by (\ref{singleexponentialq0}) leads to the following expressions for the spectral functions $a$ and $b$:
\begin{align}\label{abexpexample}
a(k) = \frac{e^{i (k L+\pi  N)} (L r \cos(L r)-i (k L+\pi N) \sin(L r))}{L r}, \qquad
b(k) = -\frac{q_0 e^{i (k L+\pi N)} \sin(L r)}{r},
\end{align}
where $r(k)$ denotes the square root
$$r(k) = \sqrt{\Big(k + \frac{\pi N}{L}\Big)^2 - \lambda q_0^2}.$$
It follows that 
$$\Delta(k) = 2(-1)^N \cos(L r(k)).$$
Note that $a$, $b$, and $\Delta$ are entire functions of $k$ even though $r(k)$  has a branch cut.
The periodic spectrum $\mathcal{P}$ is given by the zeros of $4 - \Delta^2 = 4\sin^2(Lr)$ and consists of the two simple zeros $\lambda^\pm$ defined by
$$\lambda^\pm = \begin{cases} - \frac{\pi N}{L} \pm q_0 & \text{if $\lambda = 1$}, \\
- \frac{\pi N}{L} \pm iq_0 & \text{if $\lambda = -1$},
\end{cases}$$
as well as the infinite sequence of double zeros
$$- \frac{\pi N}{L} \pm \sqrt{\frac{n^2 \pi^2}{L^2} + \lambda q_0^2}, \qquad n \in \Z \setminus \{0\}.$$
If $\lambda = 1$, then all zeros are real; if $\lambda = -1$, then the zeros are real for $|n|  \geq L q_0/\pi$ and non-real for $|n| < L q_0/\pi$.
The function $\sqrt{4 - \Delta^2}$ is single-valued on $\C \setminus \mathcal{C}$, where $\mathcal{C}$ defined in (\ref{calCdef1}) and (\ref{calCdef2}) consists of the single branch cut $\mathcal{C} = (\lambda^-, \lambda^+)$.
We fix the branch in the definition of $r$ so that $r:\C \setminus \mathcal{C} \to \C$ is analytic and $r(k) = k + \pi N /L + O(k^{-1})$ as $k \to \infty$.
Then, using (\ref{sqrt4Delta2}) to fix the overall sign,
$$\sqrt{4 - \Delta^2} = 2(-1)^N \sin(Lr),$$
and hence the function $\tilde{\Gamma}:\C \setminus \mathcal{C} \to \C$ defined in (\ref{tildeGammadef}) is given by
\begin{align}\label{tildeGammaexpexample}
 \tilde{\Gamma}(k) = -\frac{\lambda i (k - r+ \frac{\pi  N}{L})}{ q_0}.
\end{align}
Note that $\tilde{\Gamma}(k)$ has no poles in spite of the fact that $\bar{b}$ has infinitely many zeros, see Remark \ref{polesremark}.
The set of branch points $\mathcal{B}$  is given by $\mathcal{B} = \{\lambda^-, \lambda^+\}$. 
Substituting the expressions (\ref{abexpexample}) and (\ref{tildeGammaexpexample}) for $a,b,\tilde{\Gamma}$ into the definition (\ref{tildevdef}) of the jump matrix $\tilde{v}$, it follows that
\begin{subequations}\label{tildevjexpexample}
\begin{align}
& \tilde{v}_1 = \begin{pmatrix} 
 \frac{2 \lambda  r (k-r+\frac{\pi N}{L})}{q_0^2} &
   \frac{i \lambda (k -r+\frac{\pi N}{L}) e^{-2 i (\theta - k L)}}{q_0} \\
 \frac{i (k - r+\frac{\pi N}{L}) e^{2 i (\theta - k L + L r)}}{q_0} & 
 \frac{(k + \frac{\pi N}{L}) (1-e^{2 i L r}) + r (1+e^{2 i L r})}{2 r} 
 \end{pmatrix}, &&  k \in i\R_+,
	\\ 
& \tilde{v}_2 =  \begin{pmatrix} 
 \frac{2 \lambda  r (k - r+ \frac{\pi N}{L})}{q_0^2} &
   \frac{i \lambda  (k - r+ \frac{\pi N}{L}) e^{-2 i (\theta - k L + L r)}}{q_0} \\
 \frac{i (k - r+ \frac{\pi N}{L}) e^{2 i (\theta -k L + L r)}}{q_0} & 1 
  \end{pmatrix}, && k \in \R_+ \setminus \mathcal{C},
   	\\  
& \tilde{v}_3 = \begin{pmatrix} 
 \frac{2 \lambda  r (k - r+ \frac{\pi N}{L})}{q_0^2} & 
 \frac{i \lambda  (k - r+ \frac{\pi N}{L}) e^{-2 i (\theta -k L + L r)}}{q_0} \\
 \frac{i (k - r+ \frac{\pi N}{L})e^{2 i (\theta - k L)}}{q_0} &
   \frac{(k+\frac{\pi N}{L})(1 - e^{-2 i L r}) + r (1+e^{-2 i L r})}{2 r} 
   \end{pmatrix}, && k \in i\R_-,
   	\\ 
&  \tilde{v}_4 = \begin{pmatrix}	
 \frac{2 \lambda  r (k - r+ \frac{\pi N}{L})}{q_0^2} & 
 \frac{i \lambda  (k - r+ \frac{\pi N}{L}) e^{-2 i (\theta - k L)}}{q_0} \\
 \frac{i  (k - r+ \frac{\pi N}{L}) e^{2 i (\theta - k L)}}{q_0} & 1 
 \end{pmatrix}, && k \in \R_- \setminus \mathcal{C}.
\end{align}
\end{subequations}

The contour  $\tilde{\Sigma}$ is equal to $\R \cup i\R \cup (\lambda^-, \lambda^+)$ and is oriented as in Figure \ref{expexamplefig}.
If $\lambda = 1$, then $\mathcal{C} = ( - \frac{\pi N}{L}-q_0,  - \frac{\pi N}{L}+q_0)$ so the formulation of the RH problem also involves at least one of the jump matrices $\tilde{v}_2^{\cut}$ and $\tilde{v}_4^{\cut}$.
If $\lambda = -1$, then $\mathcal{C} = ( - \frac{\pi N}{L}-iq_0,  - \frac{\pi N}{L}+iq_0)$ so the formulation involves the jump matrices $\tilde{v}^{\cut}_{D_2}$ and $\tilde{v}^{\cut}_{D_3}$ if $N \geq 1$; $\tilde{v}^{\cut}_{D_1}$ and $\tilde{v}^{\cut}_{D_4}$ if $N \leq -1$; and $\tilde{v}_1^{\cut}$ and $\tilde{v}_3^{\cut}$ if $N = 0$. 
The case $N = 0$ of constant initial data was analyzed in \cite{FLDxperiodic} and the two cases $N$  and $-N$ are related by the change of variables $x \to -x$; we therefore henceforth assume that $N \geq 1$. 

Let $N \geq 1$ and $\mathfrak{r}(k) = \sqrt{|(k + \pi N/L)^2 - \lambda q_0^2|} \geq 0$. A computation shows that if $\lambda = 1$, then
\begin{align}\nonumber
& \tilde{v}_2^{\cut} = 
\begin{pmatrix}
 0 & \frac{i(k - i \mathfrak{r}(k) + \frac{\pi N}{L}) e^{-2 i (\theta - k L)}}{q_0} \\
 \frac{i (k + i \mathfrak{r}(k) + \frac{\pi N}{L}) e^{2 i (\theta - k L)}}{q_0} & 
 e^{-2 L\mathfrak{r}(k)} 
\end{pmatrix},
&& k \in \mathcal{C} \cap \R_+,	
	\\\label{tildev24cutexpexample}
& \tilde{v}_4^{\cut} = \begin{pmatrix}
 0 & \frac{i (k +i \mathfrak{r}(k) + \frac{\pi N}{L})e^{-2 i (\theta - k L)}}{q_0} \\
 \frac{i( k - i \mathfrak{r}(k) + \frac{\pi N}{L}) e^{2 i (\theta - k L)}}{q_0} & 1 
\end{pmatrix},
&& k \in \mathcal{C} \cap \R_-,
\end{align}
where $\mathfrak{r}(k) = \sqrt{q_0^2 - (k + \pi N/L)^2} \geq 0$, while if $\lambda=-1$, then
\begin{align}\nonumber
& \tilde{v}^{\cut}_{D_2} = \begin{pmatrix}
 1 & 0 \\
 \frac{2 i \mathfrak{r}(k) e^{2 i(\theta - k L)}}{q_0} & 1 
\end{pmatrix},
&& k \in \mathcal{C} \cap D_2,	
	\\ \label{tildev13cutexpexample}
& \tilde{v}^{\cut}_{D_3} = \begin{pmatrix}
 1 & -\frac{2 i \mathfrak{r}(k) e^{-2 i (\theta - k L)}}{q_0} \\
 0 & 1 
\end{pmatrix},
&& k \in \mathcal{C} \cap D_3,
\end{align}
where $\mathfrak{r}(k) = \sqrt{q_0^2 + (k + \pi N/L)^2} \geq 0$.
We conclude that in the case of the single exponential initial profile (\ref{singleexponentialq0}) with  $N \geq 1$, the RH problem \ref{RHmtilde} for $\tilde{m}$ can be formulated as follows.

\begin{figure}
\begin{center}
\begin{overpic}[width=.45\textwidth]{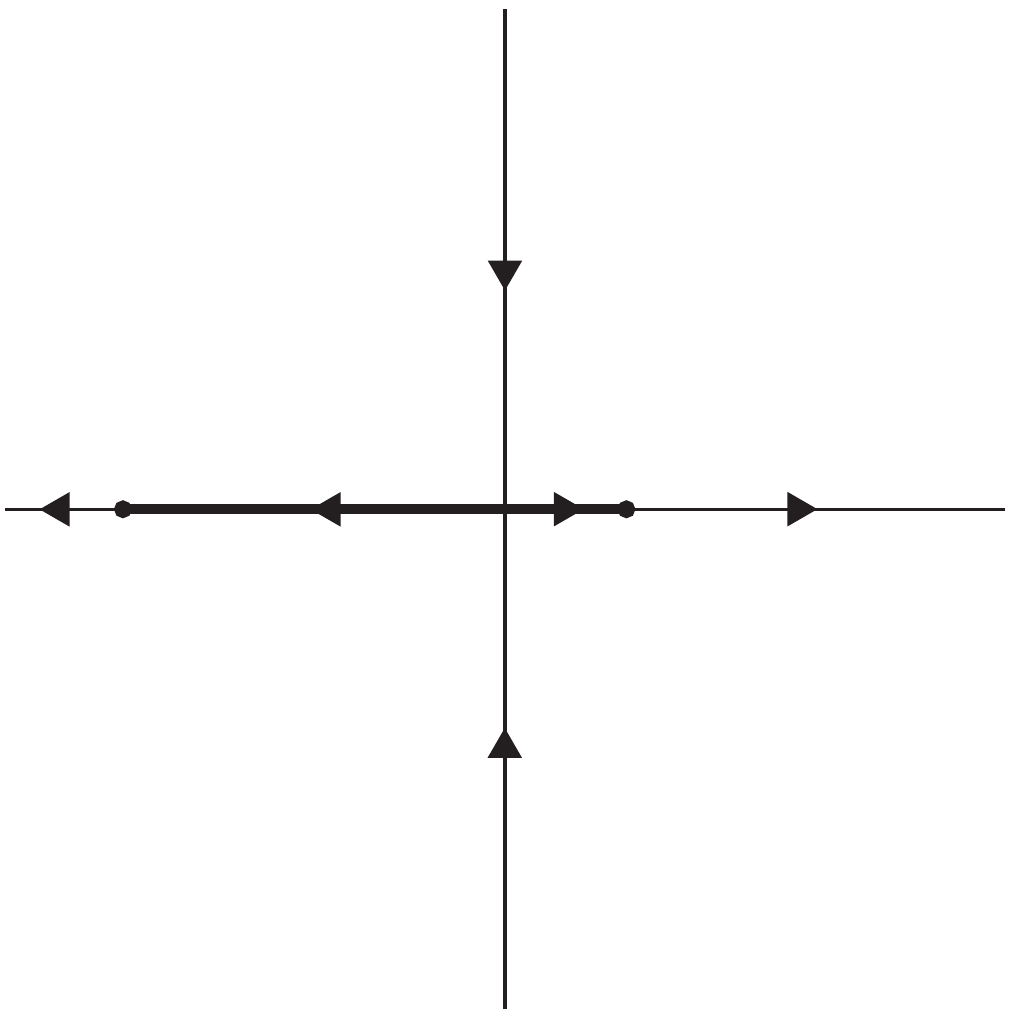}
      \put(2,43){\small $-\frac{\pi N}{L} -q_0$} 
      \put(51,43){\small $-\frac{\pi N}{L} + q_0$} 
      \put(52.5,72.5){\small $\tilde{v}_1$} 
      \put(76,53.5){\small $\tilde{v}_2$} 
      \put(54,54){\small $\tilde{v}_2^{\cut}$} 
      \put(52.5,25){\small $\tilde{v}_3$} 
      \put(4,53.5){\small $\tilde{v}_4$} 
      \put(30,54){\small $\tilde{v}_4^{\cut}$} 
     \end{overpic}\qquad
     \begin{overpic}[width=.45\textwidth]{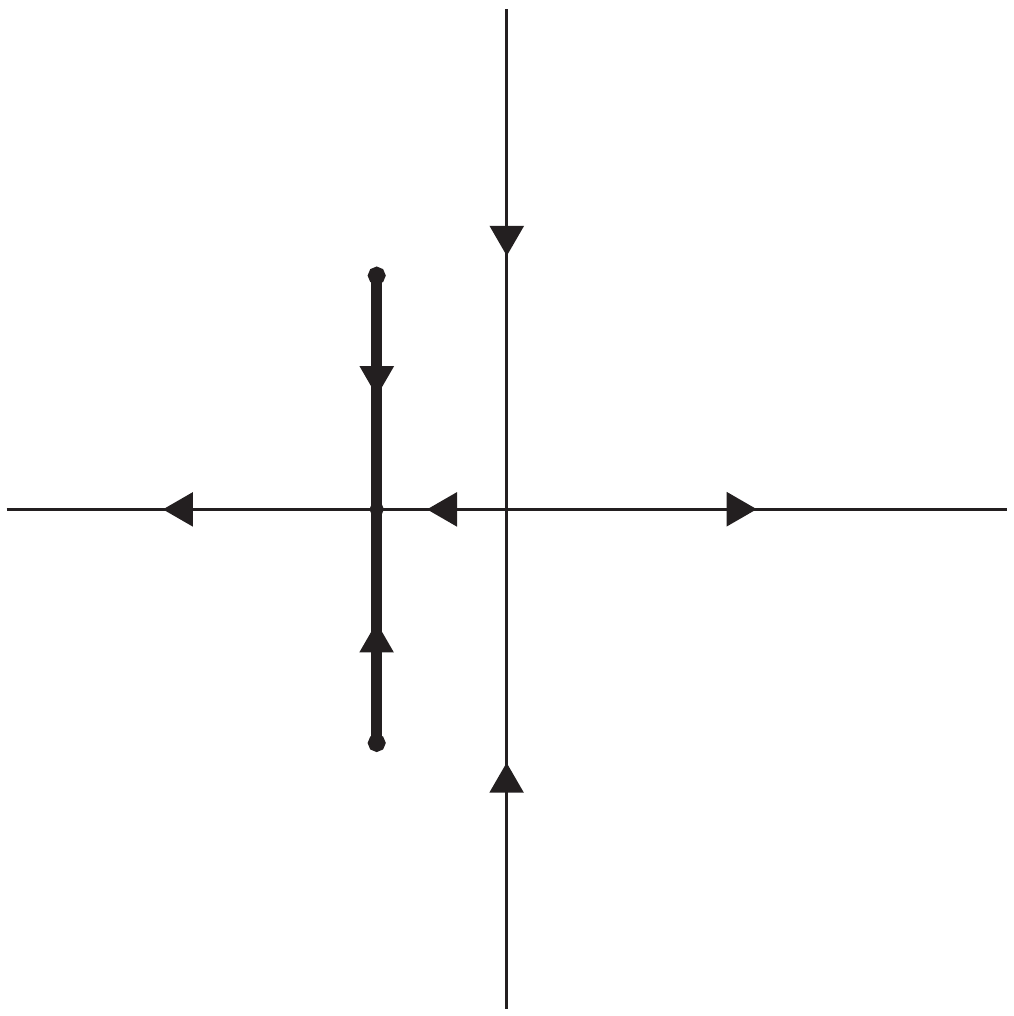}
      \put(11,25.5){\small $-\frac{\pi N}{L}-iq_0$} 
      \put(11,72.5){\small $-\frac{\pi N}{L}+iq_0$} 
      \put(52.5,75){\small $\tilde{v}_1$} 
      \put(39.4,63){\small $\tilde{v}^{\cut}_{D_2} $} 
      \put(70,53.5){\small $\tilde{v}_2$} 
      \put(52.5,22){\small $\tilde{v}_3$} 
      \put(39.4,35){\small $\tilde{v}^{\cut}_{D_3} $} 
      \put(16,53.5){\small $\tilde{v}_4$} 
      \put(42,53.2){\small $\tilde{v}_4$} 
     \end{overpic}
    \caption{\label{expexamplefig} The jump contour and jump matrices for RH problem \ref{RHmtildeexpexample} for $\lambda = 1$ (left) and $\lambda = -1$ (right).}
     \end{center}
\end{figure}
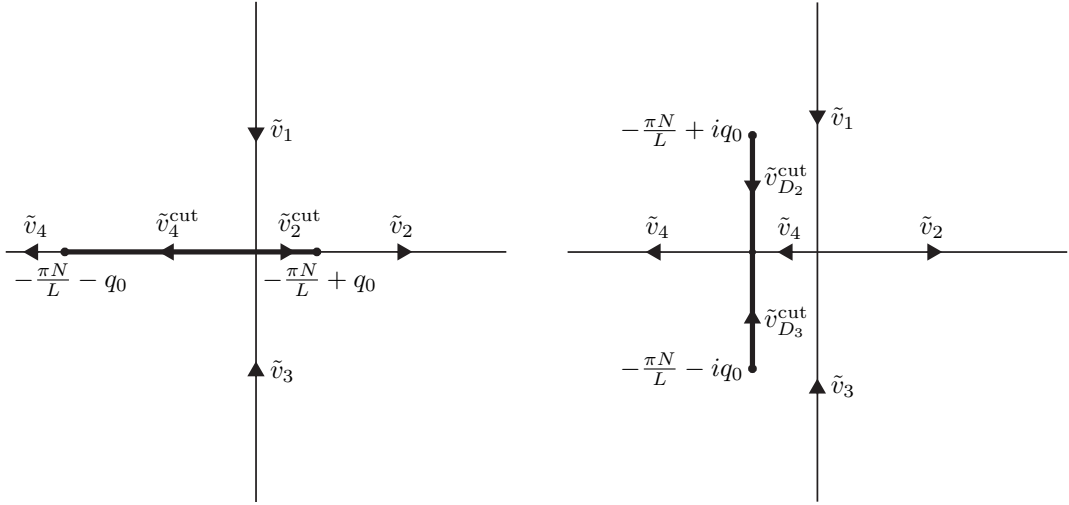

\begin{RHproblem}{\bf(The RH problem for single exponential initial datum)}\label{RHmtildeexpexample}
Find a $2 \times 2$-matrix valued function $\tilde{m}(x,t,k)$ with the following properties:
\begin{itemize}
\item $\tilde{m}(x,t,\cdot) : \C \setminus \tilde{\Sigma} \to \C^{2 \times 2}$ is analytic.

\item The limits of $m(x,t,k)$ as $k$ approaches $\tilde{\Sigma} \setminus \{0, -\frac{\pi N}{L}, \lambda^\pm\}$ from the left and right exist, are continuous on $\tilde{\Sigma} \setminus \{0, -\frac{\pi N}{L}, \lambda^\pm\}$, and satisfy
\begin{align}\label{mtildejumpexpexample}
  \tilde{m}_-(x,t,k) = \tilde{m}_+(x, t, k) \tilde{v}(x, t, k), \qquad k \in \tilde{\Sigma} \setminus \{0, -\pi N/L,  \lambda^\pm\},
\end{align}
where $\tilde{v}$ is defined as follows (see Figure \ref{expexamplefig}):
\begin{itemize}
\item If $\lambda = 1$, then
$$\tilde{v} = \begin{cases} 
\tilde{v}_1 & \text{on $i\R_+$},			\\
\tilde{v}_2 & \text{on $(-\frac{\pi N}{L} + q_0, +\infty) \cap \R_+$},	\\
\tilde{v}_2^{\cut} & \text{on $(-\frac{\pi N}{L} - q_0, -\frac{\pi N}{L} + q_0) \cap \R_+$},	\\
\tilde{v}_3 & \text{on $i\R_-$},			\\
\tilde{v}_4 & \text{on $(-\infty, -\frac{\pi N}{L}-q_0)$},	\\
\tilde{v}_4^{\cut} & \text{on $(-\frac{\pi N}{L}-q_0, -\frac{\pi N}{L} + q_0) \cap \R_-$},
\end{cases}$$
where $\{\tilde{v}_j\}_1^4$, $\tilde{v}_2^{\cut}$, and $\tilde{v}_4^{\cut}$
are defined by (\ref{tildevjexpexample}) and (\ref{tildev24cutexpexample}).

\item If $\lambda = -1$, then
$$\tilde{v} = \begin{cases} 
\tilde{v}_1 & \text{on $i\R_+$},	\\
\tilde{v}^{\cut}_{D_2}  & \text{on $(-\frac{\pi N}{L}, -\frac{\pi N}{L}+iq_0)$},	\\
\tilde{v}_2 & \text{on $\R_+$},			\\
\tilde{v}_3 & \text{on $i\R_-$},	\\
\tilde{v}^{\cut}_{D_3}  & \text{on $(-\frac{\pi N}{L}-iq_0, -\frac{\pi N}{L})$}, 	\\
\tilde{v}_4 & \text{on $\R_-$},
\end{cases}$$
where $\{\tilde{v}_j\}_1^4$, $\tilde{v}^{\cut}_{D_2}$, and $ \tilde{v}^{\cut}_{D_3}$
are defined by (\ref{tildevjexpexample}) and (\ref{tildev13cutexpexample}).

\end{itemize}

\item $\tilde{m}(x,t,k) = I  + O\big(k^{-1}\big)$ as $k \to \infty$, $k \in \C \setminus \cup_{n\in \Z} \mathcal{D}_n$.

\item $\tilde{m}(x,t,k) = O(1)$ as $k \to \{0, -\frac{\pi N}{L}, \lambda^\pm\}$, $k \in \C \setminus \tilde{\Sigma}$.

\end{itemize}
\end{RHproblem}

\begin{remark}
It is easy to verify that the jump matrices in RH problem \ref{RHmtildeexpexample}  satisfy the following consistency conditions at the origin:
\begin{align*}
\begin{cases}
 (\tilde{v}_4^{\cut} )^{-1} \tilde{v}_3 (\tilde{v}_2^{\cut} )^{-1} \tilde{v}_1\big|_{k=0} = I, & \lambda = 1,
	\\
\tilde{v}_4^{-1} \tilde{v}_3^{\cut} \tilde{v}_2^{-1} \tilde{v}_1^{\cut} \big|_{k=0} = I, & \lambda = -1.
\end{cases}
\end{align*}
\end{remark}

\subsection{Solution of the RH problem for $\tilde{m}$}
The RH problem \ref{RHmtildeexpexample} for $\tilde{m}$ can be solved explicitly by transforming it to a RH problem which has a constant off-diagonal jump across the branch cut $\mathcal{C} = (\lambda^-, \lambda^+)$.

The jump matrices $\tilde{v}_1$  and $\tilde{v}_3$ in (\ref{tildevjexpexample}) admit the factorizations
\begin{align*}
 & \tilde{v}_1 = \begin{pmatrix}
 1 & 0 \\
 \frac{i \lambda q_0 e^{2 i (\theta -k L+L r)}}{2 r} & 1   \end{pmatrix}\begin{pmatrix}
  \frac{2 \lambda  r (k - r+ \frac{\pi N}{L})}{q_0^2} &
   \frac{i \lambda  e^{-2 i (\theta - k L)} (k - r+ \frac{\pi N}{L})}{q_0} \\
 0 & \frac{k + r+ \frac{\pi N}{L}}{2 r} 
  \end{pmatrix},
  	\\
&  \tilde{v}_3 =  \begin{pmatrix}  
  1 & 0 \\
 \frac{i \lambda q_0 e^{2 i (\theta - k L)}}{2 r} & 1 
   \end{pmatrix}\begin{pmatrix}
 \frac{2 \lambda  r (k - r+ \frac{\pi N}{L})}{q_0^2} &
   \frac{i \lambda  (k - r+ \frac{\pi N}{L}) e^{-2 i (\theta - k L + L r)}}{q_0} \\
 0 & \frac{k + r+ \frac{\pi N}{L}}{2 r} 
  \end{pmatrix}.
\end{align*} 
It follows that the jump across $i\R$ can be removed by introducing a new solution $\hat{m}$ by
\begin{align}\label{expexamplehatmdef}
\hat{m} = \tilde{m} \times \begin{cases}
\begin{pmatrix}
 1 & 0 \\
 \frac{i \lambda q_0 e^{2 i (\theta -k L+L r)}}{2 r} & 1   \end{pmatrix}, & k \in D_1, 
 	\\
\begin{pmatrix}
  \frac{2 \lambda  r (k - r+ \frac{\pi N}{L})}{q_0^2} &
   \frac{i \lambda  e^{-2 i (\theta - k L)} (k - r+ \frac{\pi N}{L})}{q_0} \\
 0 & \frac{k + r+ \frac{\pi N}{L}}{2 r} 
  \end{pmatrix}^{-1}, & k \in D_2, 
 	\\
\begin{pmatrix}  
  1 & 0 \\
 \frac{i \lambda q_0 e^{2 i (\theta - k L)}}{2 r} & 1 
   \end{pmatrix}, & k \in D_3, 
 	\\
\begin{pmatrix}
 \frac{2 \lambda  r (k - r+ \frac{\pi N}{L})}{q_0^2} &
   \frac{i \lambda  (k - r+ \frac{\pi N}{L}) e^{-2 i (\theta - k L + L r)}}{q_0} \\
 0 & \frac{k + r+ \frac{\pi N}{L}}{2 r} 
  \end{pmatrix}^{-1}, & k \in D_4.
\end{cases}
\end{align}
Straightforward computations using (\ref{tildevjexpexample})-(\ref{tildev13cutexpexample}) show that $\hat{m}$  only has a jump across $(\lambda^-, \lambda^+)$.
Let us orient the contour
$$(\lambda^-, \lambda^+) = \begin{cases} (-\frac{\pi N}{L}-q_0,-\frac{\pi N}{L}+q_0), & \lambda = 1, \\ 
(-\frac{\pi N}{L}-iq_0, -\frac{\pi N}{L}+iq_0), & \lambda = -1, \end{cases}$$
to the right if $\lambda = 1$  and upward if $\lambda = -1$.
It follows that $\hat{m}(x,t,\cdot) : \C \setminus (\lambda^-, \lambda^+) \to \C^{2 \times 2}$ is analytic, that $\hat{m} = I  + O(k^{-1})$ as $k \to \infty$, and that $\hat{m}$  satisfies the jump condition
\begin{align}\label{mhatjumpexp}
\hat{m}_- = \hat{m}_+ \begin{pmatrix}  
 0 & f  \\
-1/f & 0 
  \end{pmatrix} \quad \text{for $k \in (\lambda^-, \lambda^+)$},
\end{align}  
where $f(k) \equiv f(x,t,k)$ is defined by
\begin{align}\label{expexamplefdef}
f(k) = -\frac{2 i \lambda r_+(k)}{q_0} e^{-2 i (\theta - k L)}
= \begin{cases} 
\frac{2 \sqrt{|q_0^2-(k + \frac{\pi N}{L})^2|} e^{-2 i (\theta - k L)}}{q_0}, & \lambda = 1,	\\
-\frac{2 i \sqrt{|(k + \frac{\pi N}{L})^2+q_0^2|} e^{-2 i (\theta - k L)}}{q_0}, & \lambda = -1.
\end{cases}
\end{align}

The jump matrix in (\ref{mhatjumpexp}) can be made constant (i.e., independent of $k$) by performing another transformation. Define $\delta(k) \equiv \delta(x,t,k)$ by
\begin{align*}
\delta(k) = e^{\frac{r(k)}{2\pi i} \int_{\lambda^-}^{\lambda^+} \frac{\ln f(s)}{r_+(s) (s-k)}
ds}, \qquad k \in \C \setminus (\lambda^-, \lambda^+).
\end{align*}
The function $\delta$ satisfies the jump relation $\delta_+ \delta_- = f$ on $(\lambda^-, \lambda^+)$ and $\lim_{k\to \infty} \delta(k) = \delta_\infty$, where
\begin{align*}
\delta_\infty = e^{-\frac{1}{2\pi i} \int_{\lambda^-}^{\lambda^+} \frac{\ln f(s)}{r_+(s)}ds}.
\end{align*}
Moreover, $\delta(k) = O((k-\lambda^\pm)^{1/4})$ and $\delta(k)^{-1} = O((k-\lambda^\pm)^{-1/4})$ as $k \to \lambda^\pm$.
Consequently, $\check{m} = \delta_\infty^{-\sigma_3} \hat{m} \delta^{\sigma_3}$ satisfies the following RH problem: (i) $\check{m}(x,t,\cdot) : \C \setminus (\lambda^-, \lambda^+) \to \C^{2 \times 2}$ is analytic, (ii) $\check{m} = I  + O(k^{-1})$ as $k \to \infty$, (iii) $\check{m} = O((k-\lambda^\pm)^{-1/4})$ as $k \to \lambda^\pm$, and (iv) $\check{m}$  satisfies the jump condition
$$\check{m}_- = \check{m}_+ \begin{pmatrix}  
 0 & 1  \\
-1 & 0 
  \end{pmatrix} \quad \text{for $k \in (\lambda^-, \lambda^+)$}.$$
The unique solution of this RH problem is given explicitly by
\begin{align}\label{mcheckexplicit}
\check{m} = \frac{1}{2}\begin{pmatrix} Q + Q^{-1} & i(Q - Q^{-1}) \\ -i(Q - Q^{-1}) & Q + Q^{-1} \end{pmatrix} \quad \text{with} \ \  Q(k) = \bigg(\frac{k-\lambda^+}{k - \lambda^-}\bigg)^{1/4},
\end{align}
where the branch of $Q: \C \setminus (\lambda^-, \lambda^+) \to \C$ is such that $Q \sim 1$ as $k \to \infty$. Since $\tilde{m}$ is easily obtained from $\check{m}$  by inverting the transformations $\tilde{m} \mapsto \hat{m} \mapsto \check{m}$, this provides an explicit solution of the RH problem for $\tilde{m}$. 

Using the explicit formula (\ref{mcheckexplicit}) for $\check{m}$ together with (\ref{recoverq}) we can find $q(x,t)$  for all $t$. Indeed, (\ref{recoverq}) implies
\begin{align}\label{expqfrommcheck}
q(x,t) 
= 2i\lim_{k \to \infty} k \hat{m}_{12}(x,t,k) 
= 2i \delta_\infty^{2} \lim_{k \to \infty} k \check{m}_{12}(x,t,k) 
= \frac{\lambda^+ - \lambda^-}{2} \delta_\infty^{2}.
\end{align}
In order to compute $\delta_\infty$, we note that
\begin{align}\label{expdeltainfty}
\delta_\infty = e^{c_0}
e^{\int_{\lambda^-}^{\lambda^+} \frac{s (x-L)}{\pi r_+(s)}ds}
e^{\int_{\lambda^-}^{\lambda^+} \frac{2s^2 t}{\pi r_+(s)}ds},
\end{align}
where the constant $c_0$ is given by
$$c_0 = -\frac{1}{2\pi i} \int_{\lambda^-}^{\lambda^+} \frac{\ln (- 2 i \lambda  r_+(s)) - \ln{q_0}}{r_+(s)}ds.$$
The integrals in (\ref{expdeltainfty}) involving $x - L$  and $t$ are easily computed by opening up the contour and performing a residue calculation (the only residue lies at infinity). This gives
$$\int_{\lambda^-}^{\lambda^+} \frac{s (x-L)}{\pi r_+(s)}ds
= \frac{i \pi N(x - L)}{L}, \qquad
\int_{\lambda^-}^{\lambda^+} \frac{2s^2 t}{\pi r_+(s)}ds
= - i \lambda q_0^2 t - \frac{2i\pi^2 N^2}{L^2} t.$$
If $\lambda = 1$, then the substitutions $s = -\pi N/L + \sigma$ and $\sigma = q_0 \sin\theta$ give
\begin{subequations}\label{expc0example}
\begin{align}
c_0 = \frac{1}{\pi} \int_{0}^{q_0} \frac{\ln (2\sqrt{q_0^2 - \sigma^2}) - \ln{q_0}}{\sqrt{q_0^2 - \sigma^2}}d\sigma
= \frac{1}{\pi} \int_{0}^{\pi/2} \ln (2\cos{\theta})d\theta 
= 0,
\end{align}
while, if $\lambda = -1$, then the substitutions $s = -\pi N/L + i\sigma$ and $\sigma = q_0 \sin\theta$ yield
\begin{align}\label{expc0lambda1}
c_0 = \frac{1}{\pi} \int_{0}^{q_0} \frac{\ln (- 2 i \sqrt{q_0^2 - \sigma^2}) - \ln{q_0}}{\sqrt{q_0^2 - \sigma^2}}  d\sigma
= \frac{1}{\pi} \int_{0}^{\pi/2} \ln (- 2 i \cos\theta)d\theta
= -\frac{\pi i}{4}.
\end{align}
\end{subequations}
It follows that
\begin{align}\label{deltainftyfinalexp}
\delta_\infty = e^{\frac{i \pi N(x - L)}{L}} e^{- i \lambda q_0^2 t - \frac{2i\pi^2 N^2}{L^2} t} \times \begin{cases} 1, & \lambda = 1, \\
e^{-\frac{\pi i}{4}}, & \lambda = -1.\end{cases}
\end{align}
Substituting this expression for $\delta_\infty$  into (\ref{expqfrommcheck}), we find that the solution $q(x,t)$ of (\ref{NLS}) corresponding to the initial datum $q(x,0) = q_0 e^{\frac{2i\pi N}{L} x}$ is given by 
\begin{align}\label{qexpexample}
q(x,t) = q_0 e^{\frac{2 i \pi N}{L} x} e^{- 2i \lambda q_0^2 t - \frac{4i \pi^2 N^2}{L^2} t}.
\end{align}
It is easy to verify that this $q$ indeed satisfies the correct initial value problem.

\begin{remark}[Finite-gap solutions]
We have shown that the single exponential solutions (\ref{qexpexample}) can be constructed by solving the RH problem \ref{RHmtilde} for $\tilde{m}$ directly. Associated to the solutions (\ref{qexpexample}) is the genus zero Riemann surface defined by the square root $\sqrt{4 - \Delta^2}$; we have seen that this is a two-sheeted cover of the complex plane with a branch cut along the single gap $(\lambda^-, \lambda^+)$. More generally, whenever the Riemann surface defined by $\sqrt{4 - \Delta^2}$ has finite genus (i.e., whenever the given initial condition corresponds to a finite-gap solution), we expect that a representation for the solution in terms of theta functions associated to the compact Riemann surface defined by $\sqrt{4 - \Delta^2}$ can be obtained by solving the RH problem \ref{RHmtilde}. 
\end{remark}

\paragraph{\bf Funding statement} The work of A.S.F. was supported by the EPSRC in the form of a senior fellowship. The work of J.L. was supported by the G\"oran Gustafsson Foundation, the Ruth and Nils-Erik Stenb\"ack Foundation, the Swedish Research Council, Grant No. 2015-05430, and the European Research Council, Grant Agreement No. 682537.

\medskip
\paragraph{\bf Acknowledgements} The authors are grateful to the two referees for several excellent suggestions.

\bibliographystyle{plain}
\bibliography{is}

\end{document}